\documentclass{LMCS}

\def\dOi{11(2:2)2015}
\lmcsheading%
{\dOi}
{1--26}
{}
{}
{Mar.~24, 2014}
{Apr.~17, 2015}
{}

\ACMCCS{[{\bf Theory of computation}]: Models of computation---Concurrency---Distributed computing models; Logic---Verification by
  model checking}


\usepackage{pslatex}
\usepackage[T1]{fontenc}
\usepackage{amsmath}
\usepackage{amsfonts}
\usepackage{amssymb}
\usepackage{bm}
\usepackage{mathtools}
\usepackage[pdflatex,recompilepics=false]{gastex}
\usepackage{stmaryrd} 
\usepackage[basic]{complexity}
\usepackage{ifthen}
\newboolean{short_lncs}
\usepackage{url}
\usepackage{hyperref}
\usepackage{ifpdf}
\ifpdf
        \usepackage{tikz}
         \usetikzlibrary{fit,positioning}
        \usetikzlibrary{arrows}
        \usetikzlibrary{automata}
        \usetikzlibrary{calc}
        \usetikzlibrary{decorations.pathmorphing}
\fi

\newcommand{\pjshuffle}{\|} 
\newcommand{\sizeA}{\scriptsize}
\newcommand{\sizeB}{\footnotesize}

\newcommand{\calT}{\mathcal{T}}
\newcommand{\calH}{\mathcal{H}}
\newcommand{\calP}{\mathcal{P}}

\newcommand{\Saux}{S_{\text{aux}}}

\newcommand{\PEPpcod}{\PEP^{\textnormal{partial}}_{\textnormal{codir}}}

\newcommand{\Act}{{\mathit{Act}}}
\newcommand{\Pcal}{{\mathcal{P}}}
\newcommand{\Mess}{{\mathtt{M}}}
\newcommand{\Conf}{{\mathit{Conf}}}
\newcommand{\Confre}{{\Conf_{\!\!\ttr=\epsilon}}}
\newcommand{\Ch}{{\mathtt{Ch}}}
\newcommand{\ttb}{{\mathtt{b}}}
\newcommand{\ttc}{{\mathtt{c}}}
\newcommand{\ttl}{{\mathtt{l}}}
\newcommand{\ttn}{{\mathtt{n}}}
\newcommand{\ttr}{{\mathtt{r}}}
\newcommand{\ttz}{{\mathtt{z}}}
\newcommand{\pad}{{\textit{pad}}}
\newcommand{\testc}[1]{{\ttc{:}#1}}
\newcommand{\testcp}[1]{{\ttc'{:}#1}}
\newcommand{\testl}[1]{{\ttl{:}#1}}
\newcommand{\testr}[1]{{\ttr{:}#1}}
\newcommand{\ttsharp}{{\mathtt{\#}}}
\newcommand{\wl}{\mathit{write}\_{\ttl}}
\renewcommand{\wr}{\mathit{write}\_{\ttr}}
\newcommand{\rl}{\mathit{read}\_{\ttl}}
\newcommand{\rr}{\mathit{read}\_{\ttr}}

\newcommand{\init}{{\textnormal{in}}}
\newcommand{\newinit}{{\textnormal{new}}}
\newcommand{\final}{{\textnormal{fi}}}
\newcommand{\xloop}{{\textnormal{loop}}}
\newcommand{\Even}{{\mathit{Even}}}
\newcommand{\Odd}{{\mathit{Odd}}}

\newcommand{\sol}{{\textnormal{sol}}}
\newcommand{\prx}{{\textnormal{proxy}}}
\newcommand{\rel}{{\textnormal{rel}}}
\newcommand{\los}{{\textnormal{los}}}
\newcommand{\wrlo}{{\textnormal{wrlo}}}
\newcommand{\Pre}{{\textnormal{Pre}}}
\newcommand{\up}{{\mathord{\uparrow}}}

\newcommand{\Nat}{{\mathbb{N}}} 
\renewcommand{\emptyset}{\varnothing}
\renewcommand{\epsilon}{\varepsilon} 
\renewcommand{\setminus}{\smallsetminus} 

\newcommand{\overto}[1]{\xrightarrow{\!\!#1\!\!}}
\newcommand{\step}[1]{\overto{#1}} 

\newcommand{\egdef}{\stackrel{\mbox{\begin{scriptsize}def\end{scriptsize}}}{=}}
\newcommand{\equivdef}{\stackrel{\mbox{\begin{scriptsize}def\end{scriptsize}}}{\Leftrightarrow}}

\newcommand{\inang}[1]{\langle #1 \rangle}

\newcommand{\size}[1]{{\mathopen{\mid}#1\mathclose{\mid}}}

\newcommand{\subword}{\sqsubseteq}
\newcommand{\supword}{\sqsupseteq}

\renewcommand{\subsetneq}{\varsubsetneq}

\theoremstyle{definition}\newtheorem{claim}[thm]{Claim}

\begin{document}



\title[On Reachability for Unidirectional Channel Systems]{
On Reachability for Unidirectional Channel Systems Extended with Regular Tests\rsuper*}


\author[P.~Jan\v{c}ar et al.]{Petr Jan\v{c}ar\rsuper a}
\address{{\lsuper a}FEI, Techn. Univ. Ostrava}
\email{Petr.Jancar@vsb.cz}
\thanks{{\lsuper a}Supported by the project GA\v{C}R:P202/11/0340.}

\author[P.~Karandikar]{Prateek Karandikar\rsuper b}
\address{{\lsuper b}Chennai Mathematical Institute and LSV, ENS Cachan}
\email{prateek@cmi.ac.in}
\thanks{{\lsuper b}Partially funded  by Tata Consultancy Services.}

\author[Ph.~Schnoebelen]{Philippe Schnoebelen\rsuper c}
\address{{\lsuper c}LSV, ENS Cachan, CNRS}
\email{phs@lsv.ens-cachan.fr}
\thanks{{\lsuper c}Supported by Grant ANR-11-BS02-001.}

\subjclass{Theory of Computation/Models of Computation/Distributed computing models; Theory of Computation/Logic/Verification by model checking}
\titlecomment{{\lsuper*}A preliminary version of this article appeared in the
  proceedings of the 7th IFIP International Conference on Theoretical
  Computer Science (IFIP-TCS 2012)~\cite{JKS-tcs2012}.}

\begin{abstract}
``Unidirectional channel systems'' (Chambart \& Schnoebelen, CONCUR 2008)
  are finite-state systems where one-way communication from a Sender to a Receiver goes
  via one reliable and one unreliable unbounded fifo channel. While
  reachability is decidable for these systems, equipping them
  with the possibility of testing regular properties on the
  contents of channels makes it undecidable. Decidability is
  preserved when only emptiness and nonemptiness tests are considered: the
  proof relies on an elaborate reduction to
  a generalized version of Post's Embedding Problem.
\end{abstract}



\maketitle


\section{Introduction}
\label{sec-intro}

\emph{Channel systems}
are a family of computational models where concurrent agents
communicate via (usually unbounded) fifo communication
channels~\cite{brand83}. They are sometimes called \emph{queue
  automata} when there is only one finite-state agent using the
channels as fifo memory buffers. These models are well-suited to the
formal specification and algorithmic analysis of communication
protocols and concurrent
programs~\cite{boigelot99b,bouajjani99b,muscholl2010}.

A particularly interesting class of channel systems are the
\emph{lossy channel systems}, ``LCSes'' for short, popularized by
Abdulla, Bouajjani, Jonsson, Finkel, \textit{et
  al.}~\cite{cece95,abdulla96b,abdulla-forward-lcs}. Lossy channels
are unreliable and can lose messages nondeterministically and without
any notification. This weaker model is easier to analyse: safety,
inevitability and several more properties are decidable for
LCSes~\cite{cece95,abdulla96b,ABRS-icomp,BBS-tocl} while they are
undecidable when channels are reliable.

Let us stress that LCSes also are an important and fundamental
computation model \textit{per se}. During the last decade, they have
been used as an automaton model to prove the decidability (or the
hardness) of problems on Timed Automata, Metric Temporal Logic, modal
logics,
etc.~\cite{abdulla-icalp05,ouaknine2006,kurucz06,konev06,lasota2008,BMOSW-fac2012,lazic2013,barcelo2013}.
They also are a very natural low-level computational model that
captures some important complexity classes in the ordinal-recursive
hierarchy~\cite{CS-lics08,SS-icalp11,KS-fossacs2013,SS-concur13,schmitz2013}.

\medskip

\noindent
\emph{Unidirectional channel systems,} ``UCSes'' for short, are channel systems
where a Sender process communicates to a Receiver process via
\emph{one reliable} and \emph{one lossy} channel, see
Fig.~\ref{fig-example-ucst}. They were introduced by Chambart and Schnoebelen who
identified them as a minimal setting to which one can reduce reachability problems
for more complex combinations of lossy and reliable channels~\cite{CS-concur08}.
\begin{figure}[htbp]
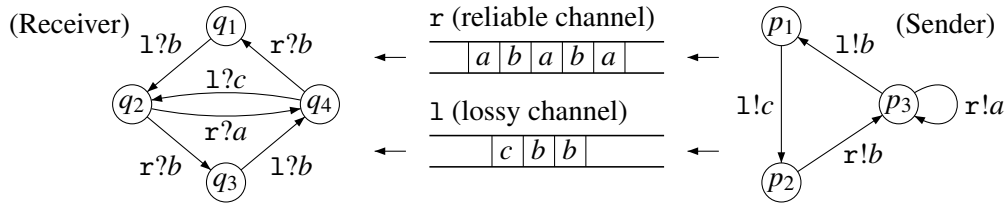

\centering
{\setlength{\unitlength}{1.04mm}
\begin{gpicture}(124,24)(-12,-2)

\gasset{ELside=r,ilength=4,flength=4,Nw=5,Nh=5,Nmr=999,loopdiam=5}

\node(q1)(14,20){$q_1$}
\node(q2)(2,10){$q_2$}
\node(q3)(14,0){$q_3$}
\node(q4)(26,10){$q_4$}
\node[Nframe=n](name)(-6,20){(Receiver)}
 \drawedge(q1,q2){$\ttl?b$}
 \drawedge(q2,q3){$\ttr?b$}
 \drawedge(q3,q4){$\ttl?b$}
 \drawedge(q4,q1){$\ttr?b$}
{\gasset{curvedepth=-1.5,ELside=r,ELdist=0.6}
 \drawedge(q2,q4){$\ttr?a$}
 \drawedge(q4,q2){$\ttl?c$}
}

\node(p1)(85,20){$p_1$}
\node(p2)(85,0){$p_2$}
\node(p3)(100,10){$p_3$}
\node[Nframe=n](name)(106,20){(Sender)}
 \drawedge(p1,p2){$\ttl!c$}
 \drawedge[ELdist=0.2,ELpos=60](p2,p3){$\ttr!b$}
 \drawedge(p3,p1){$\ttl!b$}
 \drawloop[ELside=l,loopangle=0](p3){$\ttr!a$}

{\gasset{AHnb=0,Nframe=n}

 {\gasset{linewidth=0.2}
 \drawline(40,2)(70,2)
 \drawline(40,6)(70,6)
 \drawline(40,14)(70,14)
 \drawline(40,18)(70,18)}

\put(40,20){$\ttr\text{ (reliable channel)}$}
\put(40,8){$\ttl\text{ (lossy channel)}$}

 \drawline(45,14)(45,18) \put(46,15){$a$}
 \drawline(49,14)(49,18) \put(50,15){$b$}
 \drawline(53,14)(53,18) \put(54,15){$a$}
 \drawline(57,14)(57,18) \put(58,15){$b$}
 \drawline(61,14)(61,18) \put(62,15){$a$}
 \drawline(65,14)(65,18)

 \drawline(48,2)(48,6) \put(49,3){$c$}
 \drawline(52,2)(52,6) \put(53,3){$b$}
 \drawline(56,2)(56,6) \put(57,3){$b$}
 \drawline(60,2)(60,6)

 \drawline[AHnb=1](37,16)(33,16)
 \drawline[AHnb=1](37,4)(33,4)
 \drawline[AHnb=1](77,16)(73,16)
 \drawline[AHnb=1](77,4)(73,4)
}

\end{gpicture}}
\caption{UCS = buffered one-way communication via one reliable and one lossy channels}
\label{fig-example-ucst}
\end{figure}

%
UCSes are limited to one-way communication: there are no channels
going from Receiver to Sender. One-way communication appears, e.g., in
half-duplex protocols~\cite{ibarra2003} or in the acyclic networks of~\cite{latorre2008,atig2008}.

The reachability problem for UCSes is quite challenging: it was proved
decidable by reformulating it more abstractly as the
\emph{(Regular) Post Embedding Problem} ($\PEP$),  which is easier to
analyze~\cite{CS-fsttcs07,CS-omegapep,CS-icalp2010}. We want to stress
that, while $\PEP$ is a natural variant of Post's Correspondence Problem,
it was first identified through questions on UCSes. Recently,
$\PEP$ has proved useful in other areas:
graph logics for databases~\cite{barcelo2013} and fast-growing complexity~\cite{KS-fossacs2013}.

\subsubsection*{Testing channel contents}
In basic channel systems, the agents are not allowed to inspect the contents of the
channels. However, it is sometimes useful to enrich the basic setup
with tests. For example, a multiplexer process will check each of its
input channels in turn and will rely on emptiness and/or non-emptiness
tests to ensure that this round robin policy does not block when one
input channel is empty~\cite{rosier86}. In other settings, channel
systems with insertion errors becomes more expressive when emptiness
tests are allowed~\cite{BMOSW-fac2012}.

In this article we consider such emptiness and non-emptiness tests, as
well as more general tests given by arbitrary regular predicates on
channel contents. A simple example is given below in
Fig.~\ref{fig-reduc1} (see page~\pageref{fig-reduc1}) where some of
Sender's actions depend on the parity of the number of messages
currently in $\ttr$. When verifying plain UCSes, one can reorder steps
and assume a two-phase behaviour where all Sender
steps occur before all Receiver steps. When one has tests,
one can no longer assume this.

\subsubsection*{Our contribution}
We extend UCSes with the possibility of testing channel contents with
regular predicates (Section~\ref{sec-ucsdef}). This makes reachability
undecidable even with restricted sets of simple tests
(Section~\ref{sec-undec}). Our main result (Theorem~\ref{thm-main}) is
that reachability is decidable for UCSes extended with emptiness and
non-emptiness tests. The proof goes through a series of reductions,
some of them nontrivial, that leave us with UCSes extended by only
emptiness tests on a single side of a single channel, called
``$Z_1^\ttl$ tests'' (sections~\ref{sec-reduc}
and~\ref{sec-ucsz-to-zl}). This minimal extension is then reduced
(Section~\ref{sec-ucszl}) to $\PEPpcod$, or ``$\PEP$ with partial
codirectness'', a nontrivial extension of $\PEP$ that was recently
proved decidable~\cite{KS-msttocs}. This last reduction extends the
reduction from UCS to $\PEP$ in~\cite{CS-omegapep}. Finally,
Section~\ref{sec-undec2} proves that emptiness and/or non-emptiness
tests strictly enrich the basic UCS model.

\subsubsection*{Related work.}
Emptiness and non-emptiness tests have been considered already
in~\cite{rosier86}, while Promela (SPIN's input language) offers head
tests (that test the first available message without consuming
it)~\cite{holzmann91}. Beyond such \emph{specific} tests, we are not
aware of results that consider models with a general notion of tests
on channel contents (except in the case of LCSes where very general
tests can be allowed without compromising the main decidability
results, see~\cite[sect.~6]{BS-fmsd2013}).

Regarding unidirectional channels, the decidability results
in~\cite{atig2008,latorre2008,heussner2012,heussner2012b,clemente2013}
apply to systems where communication between two agents is limited to
a \emph{single} one-way channel (sometimes complemented with a finite
shared memory, real-time clock, integer-valued counter, or local
pushdown stack). Finally let us mention the recent work by Clemente \textit{et al}.\
where fifo and ``bag'' channels can be mixed: one can see bag channels as
unreliable channels where the temporal ordering of messages is not preserved~\cite{clemente2014}.



\section{Unidirectional Channel Systems}
\label{sec-ucsdef}

\subsection{Unidirectional Channel System with Tests}
A \emph{UCST} is a tuple
$S=(\Ch, \Mess, Q_1, \Delta_1, Q_2, \Delta_2)$, where $\Mess$ is the finite
alphabet of \emph{messages}, $Q_1$, $Q_2$ are the disjoint
finite sets of \emph{states} of Sender and Receiver, respectively,
and $\Delta_1$, $\Delta_2$ are the finite sets  of \emph{rules} of
Sender and Receiver, respectively.
$\Ch=\{\ttr,\ttl\}$ is a fixed set of channel names, just
\emph{channels} for short, where $\ttr$
is \emph{reliable} and $\ttl$ is \emph{lossy} (since messages
in $\ttl$ can spontaneously disappear).

A rule $\delta\in\Delta_i$ is a tuple $(q,\ttc,\alpha,q')\in Q_i\times
\Ch\times \Act\times Q_i$ where the set of actions $\Act$ contains
\emph{tests}, checking whether the contents of $\ttc\in\Ch$ belongs to
some regular language $R\in\Reg(\Mess)$, and \emph{communications}
(sending a message $a\in\Mess$ to $\ttc$ in the case of Sender's actions, reading it for
Receiver's). Allowed actions also include the \emph{empty action} (no test, no
communication) that will be treated as ``sending/reading the empty
word $\epsilon$''; formally we put $\Act\egdef \Reg(\Mess)\cup \Mess\cup\{\epsilon\}$.

We also write a rule  $(q,\ttc,\alpha,q')$ as
$q\step{\ttc,\alpha}q'$, or specifically
$q\step{\testc{R}}q'$ for a rule where the action is a
test on $\ttc$, and $q\step{\ttc!a}q'$ or $q\step{\ttc?a}q'$ when the action
is a communication by Sender or by Receiver, respectively.
We also write just $q\step{}q'$ or  $q\step{\top}q'$ when the action is
empty.

In graphical representations like Fig.~\ref{fig-example-ucst}, Sender and
Receiver are depicted as two disjoint directed graphs, where states appear
as nodes and where rules $q\step{\ttc,\alpha}q'$ appear as edges from $q$ to
$q'$ with the corresponding labellings.

\subsection{Operational Semantics}
The behaviour of a UCST is defined via an operational semantics
along
standard lines. A \emph{configuration} of $S=(\Ch, \Mess, Q_1, \Delta_1, Q_2,
\Delta_2)$ is a tuple $C\in \Conf_S\egdef Q_1\times Q_2\times
\Mess^*\times \Mess^*$. In $C=(q_1,q_2,u,v)$, $q_1$ and $q_2$ are the
current states of Sender and Receiver, respectively, while $u$ and
$v$ are the current contents of $\ttr$ and $\ttl$, respectively.

The rules in $\Delta_1\cup\Delta_2$ give rise to transitions in the expected way. We use two
notions of transitions, or ``steps'', between configurations. We start with so-called
``reliable'' steps: given two configurations $C=(q_1,q_2,u,v)$,
$C'=(q'_1,q'_2,u',v')$ and a rule $\delta=(q,\ttc,\alpha,q')$, there is a
reliable step denoted $C\step{\delta} C'$ if, and only if, the
following four conditions are satisfied:
\begin{description}
\item[states] $q=q_1$ and $q'=q'_1$ and $q_2=q'_2$ (for Sender rules), or
$q=q_2$ and $q'=q'_2$ and $q_1=q'_1$ (for Receiver rules);
\item[tests]
if $\delta$ is a test rule $q\step{\testc{R}}q'$, then $\ttc=\ttr$ and $u\in R$, or
$\ttc=\ttl$ and $v\in R$, and furthermore $u'=u$ and $v'=v$;
\item[writes]
if $\delta$ is a writing rule $q\step{\ttc!x}q'$ with
$x\in\Mess\cup\{\epsilon\}$, then $\ttc=\ttr$ and $u'=u\, x$
and $v'=v$, or $\ttc=\ttl$ and $u'=u$ and $v'=v\, x$;
\item[reads]
if $\delta$ is a reading rule $q\step{\ttc?x}q'$, then $\ttc=\ttr$ and $u=x\, u'$
and $v'=v$, or $\ttc=\ttl$ and $u'=u$ and $v=x\, v'$.
\end{description}

\noindent This reliable behaviour is completed with message losses. For
$v,v'\in\Mess^*$, we write $v'\subword_1 v$ when $v'$ is obtained by
deleting a single (occurrence of a) symbol from $v$, and we let
$\subword$ denote the reflexive-transitive closure of $\subword_1$.
Thus $v'\subword v$ when $v'$ is a scattered subword, i.e., a
subsequence, of $v$.
(E.g., $aba\subword_1 abba$ and  $aa\subword abba$.)
This is extended to configurations and we write
$C'\subword_1 C$ or $C' \subword C$ when $C'=(q_1,q_2,u,v')$ and
$C=(q_1,q_2,u,v)$ with $v'\subword_1 v$ or $v' \subword v$,
respectively.
Now, whenever $C'\subword_1 C$, the operational semantics of $S$
includes a step from $C$ to $C'$, called a \emph{message loss} step,
and denoted $C\step{\los}C'$, considering that ``$\los$'' is an extra,
implicit rule that is always allowed.

Thus a step $C\step{\delta}C'$ of $S$ is either a reliable step, when
$\delta\in\Delta_1\cup\Delta_2$, or a (single) message loss, when
$\delta=\los$.

\begin{rem}[On reliable steps]
As is usual with unreliable channel systems, the reliable semantics
plays a key role even though the object of our
study is reachability via not necessarily reliable steps. First
it is a normative yardstick from which one defines the unreliable
semantics by extension.
Then many hardness results on lossy systems are proved
via reductions where a lossy system simulates in some way the
reliable (and Turing-powerful) behaviour: proving the correctness of
such reductions requires having
the concept of reliable steps.
\qed
\end{rem}

\begin{rem}[UCSTs and well-structured systems]
\label{remark-higman-and-Conf-subword}
It is well-known that $(\Mess^*,\subword)$ is a well-quasi-order (a
wqo): any infinite sequence $v_0,v_1,v_2,\ldots$ of words over $\Mess$
contains an infinite increasing subsequence $v_{i_0}\subword
v_{i_1}\subword v_{i_2} \subword \cdots$ This classic result, called
Higman's Lemma,  plays a fundamental role in the algorithmic
verification of lossy channel systems and other well-structured
systems~\cite{cece95,finkel98b}.
Here we
note that $(\Conf,\subword)$ is \emph{not} a wqo since $C\subword D$
requires equality on channel $\ttr$, so that UCSTs are not
well-structured systems despite the presence of a lossy channel.
\qed
\end{rem}

\subsection{Reachability}

A \emph{run} from $C_0$ to $C_n$ is a sequence of chained steps
$C_0\step{\delta_1}C_1\step{\delta_2}C_2\cdots\step{\delta_n}C_n$,
abbreviated as $C_0\step{*}C_n$ (or $C_0\step{+}C_n$ when we rule
out zero-length runs).

The \emph{(Generalized) Reachability Problem}, or just ``G-G-Reach''
for short, is the question, given a UCST $S=(\Ch, \Mess, Q_1, \Delta_1, Q_2, \Delta_2)$, some states
$p_\init,p_\final\in Q_1$, $q_\init,q_\final\in Q_2$, some regular
languages $U, V, U', V' \in\Reg(\Mess)$, whether there are some
$u\in U$, $v\in V$, $u'\in U'$ and $v'\in V'$ such that $S$ has a
run $C_\init =(p_\init,q_\init,u,v) \step{*}
C_\final=(p_\final,q_\final,u',v')$.

Since $U$, $V$, $U'$, $V'$ can be taken as singleton sets, the
G-G-Reach problem is more general than asking whether $S$ has a run
$C_\init\step{*}C_\final$ for some given initial and final
configurations. We shall need the added generality in
Section~\ref{sec-ucsz-to-zl} in particular.
However, sometimes we will also need to put
restrictions on $U$, $V$, $U'$, $V'$. We use E-G-Reach to
denote the reachability problem where $U=V=\{\epsilon\}$, i.e., where
$C_\init$ has empty channels (E is for ``Empty''), while $U',V'\in\Reg(\Mess)$ are not constrained. We will also
consider the E-E-Reach restriction where $U=V=U'=V'=\{\epsilon\}$.
It is known ---see \cite[Theo~3.1]{CS-concur08}--- that E-E-Reach is
decidable for UCSes, i.e., UCSTs that do not use tests.



\section{Testing channels and the undecidability of reachability}
\label{sec-undec}

Despite their similarities, UCSes and LCSes (lossy channel systems)
behave differently. The algorithms deciding reachability for LCSes can
easily accommodate regular (or even more expressive)
tests~\cite[Sect.~6]{BS-fmsd2013}. By contrast, UCSes become
Turing-powerful when equipped with regular tests. The main result of
this section is the undecidability of reachability  for
UCSTs.
To state the respective theorem in a stronger version, we first
introduce a notation for restricting the (regular) tests.

\subsection{Restricted sets of tests}

When
$\calT\subseteq\Reg(\Mess)$, we write UCST[$\calT$] to denote the class of
UCSTs where only tests, i.e. languages,
belonging to $\calT$ are allowed. Thus UCSTs and
UCSes coincide with UCST[$\Reg(\Mess)$] and UCST[$\emptyset$],
respectively.
We single out some simple tests
(i.e., languages) defined via regular expressions:
\begin{xalignat*}{5}
	\Even&\egdef(\Mess.\Mess)^*, &
	\Odd&\egdef\Mess.\Even, &
	Z&\egdef\epsilon, &
N&\egdef\Mess^+, &
H_a&\egdef a.\Mess^*.
\end{xalignat*}
Thus $\calP=\{\Even,\Odd\}$ is the set of \emph{parity} tests,
$Z$ is the \emph{emptiness} (or ``zero'') test, $N$ is the
\emph{non-emptiness} test and $\calH=\{H_a\mid a\in\Mess\}$
is the set of \emph{head} tests (that allows
checking what is the first message in a channel \emph{without consuming it}).
Note that the non-emptiness test
can be simulated with head tests.

Before proving (in later sections) the decidability of G-G-Reach for
UCST[$\{Z,N\}$], we start by showing that E-E-Reach is undecidable for
both UCST[$\calP$] and UCST[$\calH$]:
this demonstrates that we get undecidability not only with simple
``global'' tests (parity tests) whose
outcome depends on the entire contents of a
channel, but also with simple ``local'' tests (head tests).

In fact, we even show the
stronger statement that E-E-Reach is undecidable for UCST[$\calP^\ttr_1$] and
UCST[$\calH_1^\ttr$], where the use of subscripts and/or superscripts
means that we consider restricted systems where only Sender (for
subscript $1$, only Receiver for subscript $2$) may use the tests, and
that the tests may only apply on channel $\ttr$ or $\ttl$ (depending on the
superscript). E.g., in UCST[$\calP^\ttr_1$] the only allowed tests are
parity tests performed by Sender on channel $\ttr$.

\begin{thm}
\label{thm-UCST-P1r-undec}
Reachability (E-E-Reach) is undecidable for both UCST[$\calP_1^\ttr$]
and UCST[$\calH_1^\ttr$].
\end{thm}

\noindent We now proceed to prove Theorem~\ref{thm-UCST-P1r-undec}
by simulating queue automata with UCSTs.

\subsection{Simulating queue automata}
\label{ssec-simulating-queue}
Like queue automata,
UCSes  have a reliable channel but, unlike them,  Sender (or Receiver)
cannot both read \emph{and} write from/to it. If Sender could somehow
read from the head of $\ttr$, it would be as
powerful as a queue automaton, i.e., Turing-powerful. Now we show that
parity tests used by
Sender on $\ttr$ allow us to construct
a simple protocol making Receiver act as a proxy for Sender and implement
read actions on its behalf.
See Fig.~\ref{fig-reduc1} for an illustrating example of how Sender
simulates a rule $p_1\step{\ttr?a}p_2$.
\begin{figure}[htbp]
\centering
{\setlength{\unitlength}{1.04mm}
\begin{gpicture}(113,26)(4,-3)

\gasset{ilength=4,flength=4,Nw=5,Nh=5,Nmr=999,loopdiam=5}

\node[Nw=9](q1)(17,13){$q_\prx$}
\node[Nw=2,Nh=2](q2)(5,20){}
\node[Nw=2,Nh=2](q3)(17,-2){}
\node[Nw=2,Nh=2](q4)(29,20){}
{\gasset{curvedepth=-1,ELside=r}
 \drawedge[ELpos=63,ELdist=0.3](q1,q2){$\ttl?a$}
 \drawedge[ELpos=35,ELdist=0.2](q2,q1){$\ttr?a$}
 \drawedge[ELpos=55,ELdist=0.5](q1,q3){$\ttl?c$}
 \drawedge[ELpos=45,ELdist=0.5](q3,q1){$\ttr?c$}
 \drawedge[ELpos=65,ELdist=0.2](q1,q4){$\ttl?b$}
 \drawedge[ELpos=35,ELdist=0.2](q4,q1){$\ttr?b$}
}

\node(p1)(102,19){$p_1$}
\node(p2)(102,1){$p_2$}
\node[Nw=2,Nh=2](pp0)(117,19){}
\node[Nw=2,Nh=2](pp1)(117,1){}
\node[Nw=2,Nh=2](pp2)(86,19){}
\node[Nw=2,Nh=2](pp3)(86,1){}
\drawedge[ELside=l,ELpos=60](p1,pp0){$\testr{\Odd}$}
\drawedge[ELside=l](pp0,pp1){$\ttl!a$}
\drawedge[ELside=l,ELpos=40](pp1,p2){$\testr{\Even}$}
\drawedge[ELside=r,ELpos=56](p1,pp2){$\testr{\Even}$}
\drawedge[ELside=r](pp2,pp3){$\ttl!a$}
\drawedge[ELside=r,ELpos=43](pp3,p2){$\testr{\Odd}$}

{\gasset{AHnb=0,Nframe=n}

 {\gasset{linewidth=0.2}
 \drawline(42,2)(72,2)
 \drawline(42,6)(72,6)
 \drawline(42,14)(72,14)
 \drawline(42,18)(72,18)}

\put(42,20){\makebox(0,0)[l]{channel $\ttr$ (reliable)}}
\put(42,8){\makebox(0,0)[l]{channel $\ttl$ (lossy)}}

 \drawline(47,14)(47,18) \put(48,15){$a$}
 \drawline(51,14)(51,18) \put(52,15){$b$}
 \drawline(55,14)(55,18) \put(56,15){$c$}
 \drawline(59,14)(59,18) \put(60,15){$a$}
 \drawline(63,14)(63,18) \put(64,15){$c$}
 \drawline(67,14)(67,18)

 \drawline(53,2)(53,6) \put(54,3){$a$}
 \drawline(57,2)(57,6)

 \drawline[AHnb=1](39,16)(35,16)
 \drawline[AHnb=1](39,4)(35,4)
 \drawline[AHnb=1](79,16)(75,16)
 \drawline[AHnb=1](79,4)(75,4)
}

\end{gpicture}}
\caption{Sender simulates ``$p_1\step{\ttr?a}p_2$'' with parity tests and proxy Receiver}
\label{fig-reduc1}
\end{figure}

\noindent Described informally, the protocol is the following:
\begin{enumerate}
\item Channel $\ttl$ is initially empty.
\item
In order to ``read'' from $\ttr$,  Sender checks and records whether the
length of the current contents of $\ttr$ is odd or even, using a parity test on
$\ttr$.
\item It then writes on $\ttl$ the message that it wants to read  ($a$
	in the example).
\item
During this time
Receiver waits in its initial $q_\prx$ state and tries to read from
$\ttl$. When it  reads a message $a$  from $\ttl$, it understands it as a request telling
it to read $a$ from $\ttr$ on behalf of Sender. Once it has performed
this read on $\ttr$ (when $a$ really was there),
it returns to $q_\prx$ and waits for the next
instruction.
\item Meanwhile, Sender checks that (equivalently, waits until) the parity of the contents of $\ttr$
  has changed, and on detecting this change, concludes that the read was
  successful.
\item
Channel $\ttl$ is now empty and the simulation of a read by Sender is concluded.
\end{enumerate}
If no messages are lost on $\ttl$, the  protocol allows  Sender to
read on $\ttr$; if a message is lost on $\ttl$, the protocol deadlocks.
Also, Sender deadlocks if it attempts to read a message that is
not at the head of $\ttr$, in particular when $\ttr$ is empty;
i.e., Sender has to guess correctly.

Our simulation of a queue
automaton thus introduces many possible deadlocks,
but it still suffices for proving
undecidability of reachability, namely of E-E-Reach for
UCST[$\calP^\ttr_1$].

To prove undecidability for UCST[$\calH_1^\ttr$] we just modify the
previous protocol. We use two copies of the message alphabet, e.g.,
using two ``colours''. When writing on $\ttr$, Sender strictly
alternates between the two colours. If now Sender wants to read a
given letter, say $a$, from $\ttr$, it checks that an $a$ (of the
right colour) is present at the head of $\ttr$ by using $\calH_1^\ttr$
tests. It then asks Receiver to read $a$ by sending a message via
$\ttl$. Since colours alternate in $\ttr$, Sender can check (i.e.,
wait until), again using head tests, that the reading of $a$ occurred.



\section{Main theorem and a roadmap for its proof}
\label{sec-roadmap}

We will omit set-brackets in the expressions like
UCST[$\{Z,N\}$], UCST[$\{Z_1,N_1\}$],
UCST[$\{Z^\ttl_1\}$]; we thus write
UCST[$Z,N$],  UCST[$Z_1,N_1$],
UCST[$Z^\ttl_1$], etc.
We now state our main theorem:

\begin{thm}
\label{thm-main}
Reachability (G-G-Reach) is decidable for UCST[$Z,N$].
\end{thm}

Hence adding emptiness and nonemptiness tests to UCSes does not compromise
the decidability of reachability (unlike what happens with  parity or head
tests).

Our proof of Theorem~\ref{thm-main} is quite long, being composed of
several consecutive reductions, some of which are nontrivial.
A scheme of the proof is depicted in
Fig.~\ref{fig-roadmap}, and we  give a brief outline in the rest of
this section.

We first recall that the reachability problem for UCSes (i.e., for
UCST[$\emptyset$]) was shown decidable via a reduction to $\PEP$ (Post's
Embedding Problem) in~\cite{CS-omegapep}.
Relying on this earlier result (by reducing UCST[$Z,N$] to
UCST[$\emptyset$]) or extending its proof (by reducing UCST[$Z,N$] to
$\PEP$ directly) does not seem at all trivial. At some point $\PEPpcod$,
a non-trivial generalization of the basic $\PEP$ problem,
was introduced as a certain intermediate step
and shown decidable
in~\cite{KS-msttocs}.

Once it is known that $\PEPpcod$ is decidable, our proof for
Theorem~\ref{thm-main} is composed of two main parts:
\begin{enumerate}
\item
One part, given in Section~\ref{sec-ucszl}, is a reduction
of E-E-Reach for UCST[$Z_1^\ttl$] to $\PEPpcod$.
It is relatively compact, since we have found  a suitable intermediate
notion between runs of UCST[$Z_1^\ttl$] and solutions
of $\PEPpcod$.

\begin{figure}[htbp]
\begin{center}

\begin{tikzpicture}[auto, {-stealth'[scale=5]}, node distance=3em,thick]
\node (ggzn) {G-G-Reach[$Z$, $N$]};
\node[below=of ggzn] (ggz1n1) {G-G-Reach[$Z_1$, $N_1$]};
\node[below=of ggz1n1] (egz1n1) {E-G-Reach[$Z_1$, $N_1$]};
\node[below=of egz1n1] (egz1) {E-G-Reach[$Z_1$]};
\node[below=of egz1] (eez1) {E-E-Reach[$Z_1$]};
\node[right=11em of eez1] (ggz1l) {G-G-Reach[$Z_1^\ttl$]};
\node[above=4.5em of ggz1l] (eez1l) {E-E-Reach[$Z_1^\ttl$]};
\node[above=4.5em of eez1l] (peppcod) {$\PEPpcod$};

\draw (ggzn) to node{Sec.~\ref{ssec-elim-NZ2}} (ggz1n1);
\draw (ggz1n1) to node {Sec.~\ref{ssec-ggz1n1-egz1n1}} (egz1n1);
\draw (egz1n1) to node {Sec.~\ref{ssec-egz1n1-egz1}} (egz1);
\draw (egz1) to node {Sec.~\ref{ssec-egz1-eez1}} (eez1);
\draw (eez1) to node {Sec.~\ref{sec-ucsz-to-zl}} node[swap]{{\footnotesize Turing reduction}} (ggz1l);
\draw (ggz1l) to node (reuse) {reuse} (eez1l);
\draw (eez1l) to node {Sec.~\ref{sec-ucszl}} (peppcod);

\node [fit=(ggz1n1) (egz1n1) (egz1) (eez1), dotted, draw, rounded corners=1em, inner sep=0.7em] (forreuse) {};
\draw [dotted,bend right=15] (forreuse) to (reuse) {};
\end{tikzpicture}

\end{center}
\caption{Roadmap of the reductions from G-G-Reach[$Z$, $N$] to $\PEPpcod$}
\label{fig-roadmap}
\label{fig-ucst-roadmap}
\end{figure}
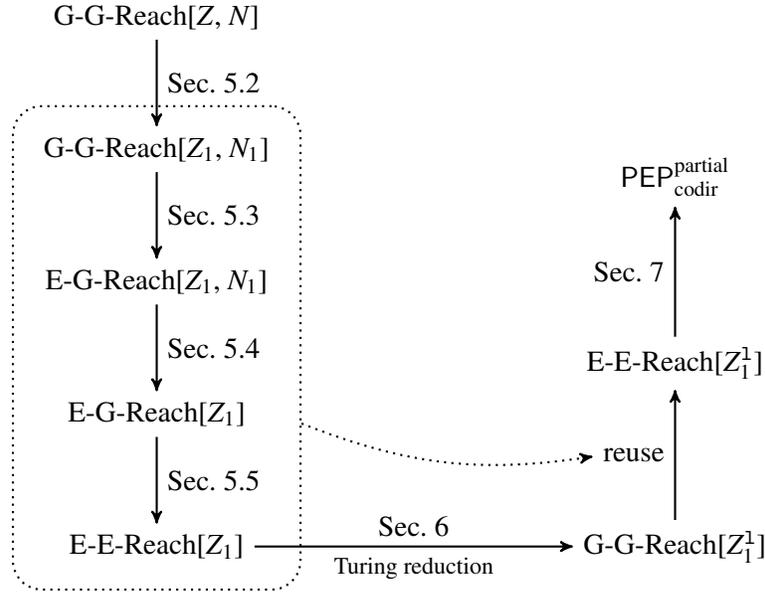

%
\item
The other part of the proof, given in sections~\ref{sec-reduc}
and~\ref{sec-ucsz-to-zl}, reduces G-G-Reach for UCST[$Z,N$] to
E-E-Reach for UCST[$Z_1^\ttl$]. It has turned out
necessary to decompose this reduction in a series of smaller steps (as
depicted in Fig.~\ref{fig-roadmap}) where features such as certain
kinds of tests, or general initial and final conditions, are
eliminated step by step. The particular way in which these features
are eliminated is important. For example, we eliminate $Z_2$ and $N_2$
tests by one simulation reducing G-G-Reach[$Z$, $N$] to
G-G-Reach[$Z_1$, $N_1$] (Sec.~\ref{ssec-elim-NZ2}); the simulation
would not work if we wanted to eliminate $Z_2$ and $N_2$ separately,
one after the other.
\end{enumerate}

\noindent One of the crucial steps in our series is the reduction from
E-E-Reach[$Z_1$] to G-G-Reach[$Z_1^\ttl$]. This is a Turing reduction,
while we otherwise use many-one reductions. Even though we
start with a problem instance where the initial and final
configurations have empty channel contents, we need oracle calls to a
problem where the initial and final conditions are more general. This
alone naturally leads to considering the G-G-Reach instances.

We note that, when UCSes are equipped with tests, reducing from
G-G-Reach to E-E-Reach is a problem in itself, for which the simple
``solution'' that we sketched in our earlier extended
abstract~\cite{JKS-tcs2012} does not work.

It seems also worth noting that all reductions in
Section~\ref{sec-reduc} treat the two channels in the same way; no
special arrangements are needed to handle the lossiness of $\ttl$. The
proofs of correctness, of course, do need to take the lossiness into
account.



\section{Reducing G-G-Reach for UCST[${Z,N}$] to E-E-Reach for UCST[${Z_1}$]}
\label{sec-reduc}

This section describes four simulations that, put together,
entail Point $1$ in Theorem~\ref{thm-ZN-reduces-to-Z1} below.
Moreover, the last three simulations also yield
Point $2$.
We note that
the simulations are tailored to the reachability problem: they may not
preserve other behavioural aspects like, e.g., termination or
deadlock-freedom.

\begin{thm}
\label{thm-ZN-reduces-to-Z1}
\hfill\\
($1$) G-G-Reach[$Z,N$] many-one reduces to E-E-Reach[$Z_1$].
\\
($2$)
G-G-Reach[$Z_1^\ttl$] many-one reduces to E-E-Reach[$Z_1^\ttl$].
\end{thm}
Before proceeding with the four reductions,
we present a simple Commutation Lemma that lets
us reorder runs and assume that they follow a specific pattern.

\subsection{Commuting steps in UCST[$\bm{Z,N}$] systems}
We say that two consecutive steps
$C\step{\delta_1}C'\step{\delta_2}C''$ (of some $S$) \emph{commute} if
$C\step{\delta_2}D\step{\delta_1}C''$ for some configuration $D$ of $S$.
The next lemma
lists some conditions that are sufficient for commuting steps in an
arbitrary UCST[$Z,N$] system $S$:
\begin{lem}[Commutation]
\label{lem-commuting-steps}
Two consecutive steps
$C\step{\delta_1}C'\step{\delta_2}C''$ commute
in any of the following cases:
\begin{enumerate}
\item \label{comm-nc} No contact: $\delta_1$ is a read/write/test by
  Sender or Receiver acting on one channel $\ttc$ (or a
  message loss on $\ttc=\ttl$), while $\delta_2$ is a
  rule of the \emph{other agent} acting on the \emph{other channel}
  (or is a loss).

\item \label{comm-pl} Postponable loss: $\delta_1$ is a message loss that
  does not occur at the head of (the current content of) $\ttl$.

\item \label{comm-as} Advanceable Sender:
$\delta_1$ is a Receiver's rule or a loss,
 and $\delta_2$ is a Sender's rule but not a $Z_1$-test.

\item \label{comm-al} Advanceable loss: $\delta_2$ is a loss
  and $\delta_1$ is not an
  ``$\testl{N}$''
  test  or a
  Sender's write on $\ttl$.
\end{enumerate}
\end{lem}
\begin{proof}
By a simple case analysis. For example, for \eqref{comm-pl} we
observe that if $\delta_1$ loses a symbol
behind the head of $\ttl$, then
there is another message at the head of $\ttl$,
and thus commuting is possible even if $\delta_2$ is
an ``$\ttl?a$'' read or an ``$\testl{Z}$'' test.
\end{proof}
We will use Lemma~\ref{lem-commuting-steps} several times and in
different ways. For the time being,  we consider
in particular the convenient restriction
to ``head-lossy'' runs. Formally, a message loss $C\step{\los} C'$ is
\emph{head-lossy} if it is of the form
$(p,q,u,av)\step{\los} (p,q,u,v)$ where $a\in\Mess$ (i.e.,
the lost message was the head of $\ttl$).
A \emph{run} $C_\init\step{*}C_\final$ is \emph{head-lossy}
if all its message
loss steps are head-lossy, or occur after all the reliable steps in
the run
(it is convenient to allow unconstrained losses at the end
of the run).
Repeated use of Point \eqref{comm-pl} in
Lemma~\ref{lem-commuting-steps} easily yields the next
corollary:
\begin{cor}
\label{coro-head-lossy-runs}
If there is a run from $C_\init$ to $C_\final$ then
there is a head-lossy run from $C_\init$ to $C_\final$.
\end{cor}

\subsection{Reducing G-G-Reach[$\bm{Z,N}$] to G-G-Reach[$\bm{Z_1,N_1}$]}
\label{ssec-elim-NZ2}

Our first reduction eliminates $Z$ and $N$ tests by Receiver. These tests are
replaced by reading two special new messages, ``$\ttz$'' and
``$\ttn$'', that Sender previously put in the channels.

Formally, we consider an instance of G-G-Reach[$Z,N$], made of a given
UCST $S=(\{\ttr,\linebreak[0] \ttl\},\linebreak[0] \Mess, Q_1, \Delta_1, Q_2, \Delta_2)$, given
states $p_\init,p_\final\in Q_1$, $q_\init,q_\final\in Q_2$, and given
languages $U,V,U',V'\in\Reg(\Mess)$.
We construct a new UCST $S'$ from $S$ as follows (see
Fig.~\ref{fig-simul3}):
\begin{enumerate}
\item
We add two special new messages $\ttz,\ttn$ to $\Mess$,
thus creating the alphabet
 $\Mess'\egdef\Mess\uplus\{\ttz,\ttn\}$.
\item
For each channel $\ttc\in\{\ttr,\ttl\}$
and each Sender's state $p\in Q_1$ we add new states $p_\ttc^1$,
$p_\ttc^2$ and
an ``\emph{(emptiness) testing loop}''
$p\step{\testc{Z}}p_\ttc^1\step{\ttc!\ttz}p_\ttc^2\step{\testc{Z}}p$
(i.e., three new rules).
\item
For every Sender's writing rule $\theta$ of the form $p\step{\ttc!x}p'$
we add a new state $p_\theta$
and the following three rules:
$p\step{\top}p_\theta$, $p_\theta\step{\ttc!\ttn}p_\theta$
(a ``\emph{padding loop}''), and $p_\theta \step{\ttc!x}p'$.
\item
For every Receiver's rule $q\step{\testc{Z}}q'$ (testing emptiness of
$\ttc$) we add the rule $q\step{\ttc?\ttz}q'$.
\item
For every Receiver's rule $q\step{\testc{N}}q''$ (testing non-emptiness of
$\ttc$) we add the rule $q\step{\ttc?\ttn}q''$.

\item
At this stage, the resulting system is called $\Saux$.

\item
Finally we remove all Receiver's tests, i.e., the  rules
$q\step{\testc{Z}}q'$ and
$q\step{\testc{N}}q''$.
We now have $S'$.
\end{enumerate}
\begin{figure}[htbp]
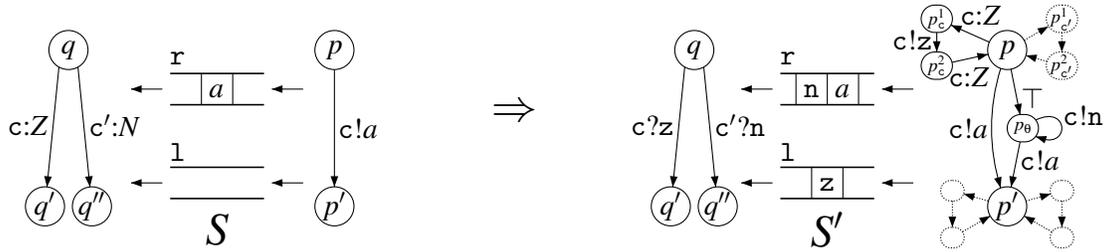

\centering
{\setlength{\unitlength}{1.04mm}
\begin{gpicture}(134,26)(-6,-3)

\gasset{ilength=4,flength=4,Nw=5,Nh=5,Nmr=999,ELdist=0.5,loopdiam=3}

\node(q1)(0,20){$q$}
\node(q2)(-3,0){$q'$}
\node(q3)(3,0){$q''$}
\drawedge[ELside=r,sxo=-1](q1,q2){$\testc{Z}$}
\drawedge[ELside=l,sxo=1](q1,q3){$\testcp{N}$}

\node(p1)(34,20){$p$}
\node(p2)(34,0){$p'$}
\drawedge[ELside=l](p1,p2){$\ttc!a$}

{\gasset{AHnb=0,Nframe=n}

 {\gasset{linewidth=0.2}
 \drawline(13,1)(25,1)
 \drawline(13,5)(25,5)
 \drawline(13,13)(25,13)
 \drawline(13,17)(25,17)}
 \node(name)(19,-3){\LARGE $S$}

\put(13,19){\makebox(0,0)[l]{channel $\ttr$}}
\put(13,7){\makebox(0,0)[l]{channel $\ttl$}}

 \drawline(17,13)(17,17) \put(18,14){$a$}
 \drawline(21,13)(21,17)

 \drawline[AHnb=1](12,15)(8,15)
 \drawline[AHnb=1](12,3)(8,3)
 \drawline[AHnb=1](30,15)(26,15)
 \drawline[AHnb=1](30,3)(26,3)
}

\node[Nframe=n](arrow)(57,12){\LARGE $\Rightarrow$}

\node(q1)(80,20){$q$}
\node(q2)(77,0){$q'$}
\node(q3)(83,0){$q''$}
\drawedge[ELside=r,sxo=-1](q1,q2){$\ttc?\ttz$}
\drawedge[ELside=l,sxo=1](q1,q3){$\ttc'?\ttn$}

\node(p1)(120,20){$p$}
\node(p2)(120,0){$p'$}
{\gasset{Nw=4,Nh=3.5}
 \node(pth)(122,10){\tiny $p_\theta$}
 \node(p11c)(111,24){\tiny $p^1_\ttc$}
 \node(p12c)(111,18){\tiny $p^2_\ttc$}
{\gasset{dash={0.15 0.2}0}
 \node(p11d)(127,24){\tiny $p^1_{\ttc'}$}
 \node(p12d)(127,18){\tiny $p^2_{\ttc'}$}
 \gasset{Nw=3,Nh=2.625}
 \node(p21c)(127,2){}
 \node(p22c)(127,-4){}
 \node(p21d)(113,2){}
 \node(p22d)(113,-4){}
 \drawedge(p1,p11d){}
 \drawedge(p11d,p12d){}
 \drawedge(p12d,p1){}
 \drawedge(p2,p21c){}
 \drawedge(p21c,p22c){}
 \drawedge(p22c,p2){}
 \drawedge(p2,p21d){}
 \drawedge(p21d,p22d){}
 \drawedge(p22d,p2){}
}
 \drawedge[ELpos=50,ELside=r,ELdist=0.5](p1,p11c){$\testc{Z}$}
 \drawedge[ELpos=45,ELside=r,ELdist=0.6](p11c,p12c){$\ttc!\ttz$}
 \drawedge[ELside=r,ELdist=0.6,ELpos=40](p12c,p1){$\testc{Z}$}
}
\drawedge[ELside=r,curvedepth=-2](p1,p2){$\ttc!a$}
\drawedge[ELpos=65,ELdist=0.3,ELside=l](p1,pth){$\top$}
\drawedge[ELside=l,ELdist=0.2,ELpos=35](pth,p2){$\ttc!a$}
\drawloop[loopangle=0,ELpos=43,ELside=l,ELdist=0.2,loopCW=y](pth){$\ttc!\ttn$}

{\gasset{AHnb=0,Nframe=n}

 {\gasset{linewidth=0.2}
 \drawline(91,1)(103,1)
 \drawline(91,5)(103,5)
 \drawline(91,13)(103,13)
 \drawline(91,17)(103,17)}
 \node(name)(97,-3){\LARGE $S'$}

\put(91,19){\makebox(0,0)[l]{channel $\ttr$}}
\put(91,7){\makebox(0,0)[l]{channel $\ttl$}}

 \drawline(93,13)(93,17) \put(94,14){$\ttn$}
 \drawline(97,13)(97,17) \put(98,14){$a$}
 \drawline(101,13)(101,17)

 \drawline(95,1)(95,5) \put(96,2){$\ttz$}
 \drawline(99,1)(99,5)

 \drawline[AHnb=1](90,15)(86,15)
 \drawline[AHnb=1](90,3)(86,3)
 \drawline[AHnb=1](108,15)(104,15)
 \drawline[AHnb=1](108,3)(104,3)
}

\end{gpicture}}
\caption{Reducing G-G-Reach[$Z,N$] to G-G-Reach[$Z_1,N_1$]: eliminating Receiver's tests}
\label{fig-simul3}
\end{figure}

%
\noindent The intuition behind $S'$ is that Sender runs a small protocol
signaling to Receiver what the status of the channels is. When a
channel is empty, Sender may write a $\ttz$ to it that Receiver can
read in place of testing for emptiness. For correctness, it is
important that Sender does not proceed any further until this $\ttz$
has disappeared from the channel. For non-emptiness tests, Sender
can always write several extraneous $\ttn$ messages before writing an
original message. Receiver can then read these $\ttn$'s in place of testing
for nonemptiness.

For $w=a_1a_2\ldots a_\ell\in\Mess^*$, we let
$\pad(w)\egdef\ttn^*a_1\ttn^*a_2\ldots\ttn^*a_\ell$ denote the set (a
regular language) of all \emph{paddings} of $w$, i.e., words obtained
by inserting any number of $\ttn$'s in front of the original messages.
Note that $\pad(\epsilon)=\{\epsilon\}$. This is extended to arbitrary
languages in the usual way: for $L\subseteq\Mess^*$,
$\pad(L)=\bigcup_{w\in L}\pad(w)$ and we note that, when $L$ is
regular, $\pad(L)$ is regular too. Furthermore, one easily derives an
FSA (a finite-state automaton) or a regular expression for $\pad(L)$ from an FSA or a regular
expression for $L$.

By replacing $S$, $U$, $V$ with $S'$, $\pad(U)$, $\pad(V)$ (and
keeping $p_\init$, $p_\final$, $q_\init$, $q_\final$, $U'$, $V'$
unchanged), the initial G-G-Reach[$Z,N$] instance is transformed into
a G-G-Reach[$Z_1,N_1$] instance. The correctness of this reduction is
captured by the next lemma, that
we immediately proceed to prove in the rest of section~\ref{ssec-elim-NZ2}:
\begin{lem}
\label{lem-corr-new-reduc}
For any $u,v,u',v'\in\Mess^*$, $S$ has a run $(p_\init, q_\init, u, v)
\step{*} (p_\final, q_\final, u', v')$ if, and only if, $S'$ has a run
$(p_\init, q_\init, \hat{u}, \hat{v})\step{*} (p_\final, q_\final, u',
v')$ for some padded words $\hat{u}\in\pad(u)$ and
$\hat{v}\in\pad(v)$.
\end{lem}
Though we are ultimately interested in $S$ and $S'$, it is convenient
to consider special runs of $\Saux$ since $\Saux$ ``contains'' both $S$
and $S'$. We rely on
Corollary~\ref{coro-head-lossy-runs} and
 tacitly assume that all runs are head-lossy. We say that a
(head-lossy) \emph{run} $C_0\step{\delta_1}C_1\step{\delta_2}\cdots
\step{\delta_n}C_n$ of $\Saux$ is \emph{faithful} if $C_0=(p_0,
q_0, u_0, v_0)$ with $u_0,v_0\in\pad(\Mess^*)$,
$C_n=(p_n,q_n, u_n, v_n)$ with $u_n,v_n\in \Mess^*$, $p_0, p_n \in Q_1$,
$q_0, q_n \in Q_2$,  and the
following two properties are satisfied (for all $i=1,2,\ldots,n$):
\begin{gather}
\tag{P1}
\begin{minipage}{0.91\textwidth}
-- if $\delta_i$ is some $p\step{\testc{Z}}p^1_\ttc$ then $\delta_{i+1}$,
$\delta_{i+2}$, and $\delta_{i+3}$ are
$p^1_\ttc\step{\ttc!\ttz}p^2_\ttc$, $q\step{\ttc?\ttz}q'$,
$p^2_\ttc\step{\testc{Z}}p$ (for some $q,q'\in Q_2$). In this case,
the subrun $C_{i-1}\step{*}C_{i+3}$ is called a \emph{P1-segment} of
the run.
\end{minipage}
\\
\tag{P2}
\begin{minipage}{0.91\textwidth}
-- if $\delta_i$ is some $p\step{\top}p_\theta$ then there is some $j>i$
such that $\delta_{i+1}, \delta_{i+2}, \ldots, \delta_{j}$ are
$p_\theta\step{\ttc!\ttn}p_\theta\step{\ttc!\ttn}
\cdots\step{\ttc!\ttn}p_\theta\step{\ttc!a}p'$ for some $a\in \Mess$
and $p'\in Q_1$. The subrun $C_{i-1}\step{*}C_j$ is called a
\emph{P2-segment}.
\end{minipage}
\end{gather}
Informally, a run is faithful if it uses the new rules (introduced in
$\Saux$) in the ``intended'' way: e.g., P1
enforces that
each $\ttz$ written by Sender (necessarily via a rule
$p_1^\ttc\step{\ttc!\ttz}p_2^\ttc$) is immediately read
after being written in the empty channel.
We note that any run of $S$ is trivially faithful
since it does not use the new rules.

We now exhibit two reversible transformations of runs of $\Saux$, one
for $Z$ tests in \textsection\ref{sssec-Z2-to-P1}, the other for $N$
tests in \textsection\ref{sssec-N2-to-P2}, that preserve faithfulness.
This will allow us to translate runs of $S$, witnessing the original
instance, to faithful runs of $S'$, witnessing the created instance,
and vice versa. Finally we show in \textsection\ref{sssec-S'-faithful}
that if there is a run of $S'$ witnessing the created instance, then
there is a faithful one as well.

When describing the two transformations we shall assume, in order to
fix notations, that we transform a test on channel $\ttl$; the case
for the channel $\ttr$ is completely analogous. For both
transformations we assume a faithful (head-lossy) run $\pi$ of $\Saux$
in the following form:
\begin{gather}
\tag{$\pi$}
(p_\init, q_\init, u_0, v_0)=
C_0\step{\delta_1}C_1\step{\delta_2}C_2\cdots\step{\delta_n}C_n=
(p_\final,q_\final, u_n, v_n)
\end{gather}
where $\delta_1,\ldots,\delta_n$ can be rules of $\Saux$ or the
``$\los$'' symbol for steps where a message is lost. For
$i=0,1,\ldots,n$, we let $C_i=(p_i,q_i,u_i,v_i)$.

\subsubsection{Trading $Z_2$ tests for P1-segments.}
\label{sssec-Z2-to-P1}
Assume that the step $C_m\step{\delta_{m+1}}C_{m+1}$ in $\pi$ is a
$Z_2$-test (an emptiness test by Receiver), hence has the form $(p,q, w, \varepsilon)
\step{\testl{Z}} (p,q',w,\varepsilon)$ if we assume $\ttc=\ttl$. We
may replace this step with the following steps
\begin{equation}\label{eq:ztransf}
(p,q, w, \varepsilon) \step{\testl{Z}}
(p^1_\ttl,q,w,\varepsilon) \step{\ttl!\ttz}
(p^2_\ttl,q,w,\ttz) \step{\ttl?\ttz}
(p^2_\ttl,q',w,\varepsilon) \step{\testl{Z}}
(p,q',w,\varepsilon)
\end{equation}
using the rules introduced in $\Saux$.
This transforms (the faithful run) $\pi$ into another faithful run
$\pi'$, decreasing the number of Receiver's tests
(by one occurrence of a $Z_2$-test).
In the other direction, if $\pi$ contains a P1-segment
$C_{m-1}\step{*}C_{m+3}$, it must be of the form
\eqref{eq:ztransf}, when the involved channel is $\ttc=\ttl$, and we
can replace it with one step
$C_{m-1}\step{\testc{Z}}C_{m+3}$, preserving faithfulness.

\subsubsection{Trading $N_2$ tests for occurrences of $\ttn$.}
\label{sssec-N2-to-P2}
Now assume that the step $C_m\step{\delta_{m+1}}C_{m+1}$ is an $N_2^\ttl$-test,
hence has the form $(p,q, u, x\,v) \step{\testl{N}}
(p,q',u,x\,v)$ for some message $x\in\Mess'$.
Now $x\not=\ttz$ since there was no $\ttz$'s in $v_0$
and, as noted above, any $\ttz$ written by Sender in a faithful run is
immediately read.
Hence $x\in \Mess\cup\{\ttn\}$. We want to replace
the $q\step{\testl{N}}q'$ test (by Receiver) with a $q\step{\ttl?\ttn}q'$ but this
requires inserting one $\ttn$ in $\ttl$, i.e., using a new rule
$p_\theta\step{\ttl!\ttn}p_\theta$ at the right moment.

We now follow the (occurrence of) $x$ singled out in $C_m$ and
find the first configuration, say $C_k$, where this $x$ appears already;
we can thus write $v_i=w_i\, x\, w'_i$, i.e.,
$C_i=(p_i,q_i,u_i,w_i \, x \, w'_i)$, for $i=k,k+1,\ldots,m$. Here $x$
always depicts the same occurrence, and e.g., $w_m\, x\, w'_m=x\, v$
entails $w_m=\epsilon$ and $w'_m=v$. By adding $\ttn$ in front of
$x$ in each $C_i$ for $i=k,k+1,\dots,m$, we obtain new
configurations $C'_k,C'_{k+1},\ldots,C'_m$ given by
$C'_i=(p_i,q_i,u_i,w_i \,\ttn\, x\, w'_i)$. Now $C'_{k}
\step{\delta_{k+1}} C'_{k+1} \step{\delta_{k+2}} \cdots
\step{\delta_m} C'_m$ is a valid run of $\Saux$ since $x$ is not read
during $C_k\step{*}C_m$ and since, thanks to the presence of $x$,
adding one $\ttn$ does not change the (non)emptiness status of $\ttl$ in this
subrun.
Moreover, since $q\step{\testl{N}}q'$ is a rule of $S$,
there is a rule $q\step{\ttl?\ttn}q'$ in $\Saux$, where $C'_m=(p,q,u,\ttn\,
x\,v)\step{\ttl?\ttn}(p,q',u,x\,v)=C_{m+1}$ is a valid step.

If $k=0$ (i.e., if $x$ is present at the beginning of $\pi$),
we have exhibited a faithful
run $C'_0 \step{*} C'_m \step{\ttl?\ttn}
C_{m+1}\step{*}C_n$, starting from $C'_0=(p_\init,q_\init,u_0,
w_0\,\ttn x\,w'_0)$, where $w_0\ttn x w'_0\in\pad(v_0)$ since
$v_0=w_0\,x\,w'_0$.
If $k>0$, the highlighted occurrence
of $x$ necessarily
appears in $C_k$ via  $\delta_k=p_{k-1}\step{\ttl!x}p_{k}$ and we
have $v_k=v_{k-1}x$.
If $\delta_k$ is a rule of $S$, we may exhibit a sequence
$C_{k-1}\step{*}C'_k$ using the new rules
\begin{gather*}
 C_{k-1} \step{\top} (p_{\delta_k},q_{k-1},u_{k-1},v_{k-1})
\step{\ttl!\ttn} (p_{\delta_k},q_{k-1},u_{k-1},v_{k-1}\,\ttn)
\step{\ttl!x} (p_k,q_{k-1},u_{k-1},v_{k-1}\,\ttn \,x)=C'_k
\:,
\end{gather*}
while if $\delta_k$ is a new rule $p_\theta\step{\ttl!x}p_k$, we can
use $C_{k-1}\step{\ttl!\ttn}\step{\ttl!x}C'_k$. In both cases we can
use $C_{k-1}\step{*}C'_k$ to construct a new faithful run
$C_0\step{*}C_{k-1}\step{*}C'_k\step{*}C'_m\step{}C_{m+1}\step{*}
C_n$.
We have again decreased the number of Receiver's tests,
now by one occurrence of an $N_2$-test.

For the backward transformation we assume that $\ttn$ occurs in a
configuration of $\pi$. We select one such occurrence and let $C_k, C_{k+1},\dots, C_m$ ($0\leq k\leq m<n$)
be the part of $\pi$
where this occurrence of $\ttn$ appears.
For $i=k,k{+}1,\ldots,m$, we
highlight this occurrence of $\ttn$ by writing $v_i$ in the form
$w_i \,\ttn\, w'_i$
(assuming w.l.o.g.\ that the $\ttn$
occurs in $\ttl$),
i.e., we write $C_i=(p_i,q_i,u_i,w_i\,\ttn\,
w'_i)$. Removing the $\ttn$ yields new
configurations $C'_k,C'_{k+1},\ldots,C'_m$ given by
$C'_i=(p_i,q_i,u_i,w_i\, w'_i)$.

We claim that
$C'_k\step{\delta_{k+1}}C'_{k+1}\cdots\step{\delta_m}C'_m$ is a valid
run of $\Saux$. For this, we only need to check that removing $\ttn$
does not make channel $\ttl$ empty in some $C'_i$
where $\delta_{i+1}$ is an $N^\ttl$-test.
If $k=0$
then $\ttn$ in $v_0=w_0\ttn w'_0$
is followed by a letter $x\in\Mess\cup\{\ttn\}$ since
$v_0\in\pad(\Mess^*)$. This $x$ remains in $\ttl$ until at least
$C_{m+1}$ since it cannot be read while $\ttn$ remains, nor can it be
lost before the $C_i\step{}C_{i+1}$ step since the run is head-lossy. If $k>0$,
then our $\ttn$ appeared in a step of the form
$C_{k-1}=(p_{\theta},q_{k-1},u_{k-1},v_{k-1})\step{\ttl!\ttn}
C_k=(p_{\theta},q_{k-1},u_{k-1},v_{k-1}\ttn)$ (for some write rule
$\theta$ of $S$, inducing $p_\theta\step{\ttl!\ttn}p_\theta$ in $\Saux$).
Since $p_0=p_\init$ is not $p_\theta$, a rule
$p_\ell\step{\top}p_\theta$ was used before step $k$, and $\pi$ has a
P2-segment $C_\ell\step{\top}\cdots C_{k-1}
\step{\ttl!\ttn}C_{k}\step{\ttl!x}\cdots C_{\ell'}$
where $\ell'\leq m$ and $x\in\Mess\cup\{\ttn\}$ is present in all
$C_{k+1},\dots,C_m$.
As before, this $x$ guarantees that
$C_{k-1}=C'_k\step{\delta_{k+1}}C'_{k+1}\cdots\step{\delta_m}C'_m$ is
a valid
run of $\Saux$.

We now recall that  $m<n$ and
that $\delta_{m+1}$ is either
$q_m\step{\ttl?\ttn}q_{m+1}$ or the loss of $\ttn$.
In the first
case, $\Saux$ has a step $C'_m\step{\testl{N}}C_{m+1}$, while in the
second case $C'_m=C_{m+1}$.

The corresponding run
$C'_0\step{*}C'_m\step{*}C_{m+1}\step{*}C_n$
in the case $k=0$, or
$C_0\step{*}C_{k-1}\step{}C'_{k+1}\step{*}C'_m\step{*}C_{m+1}\step{*}C_n$
in the case $k>0$,
is a faithful run;
we have thus removed an occurrence of $\ttn$, possibly at a
cost of introducing one $N_2$ test.

\subsubsection{Handling $S'$ runs and faithfulness.}
\label{sssec-S'-faithful}
Since a witness run of $S$ is (trivially) faithful,
the above transformations
allow us to  remove one
by one all occurrences of Receiver's $Z$ and $N$ tests, creating a
(faithful) witness run for $S'$ (with a possibly padded $C_0$).
We have thus proved the ``only-if''
part of
Lemma~\ref{lem-corr-new-reduc}.
The ``if'' part is shown analogously, now using
the two transformations in the other direction and removing
occurrences of the new $\ttz$ and $\ttn$ messages, \emph{with one
  proviso}: we only transform faithful runs.
We thus need to show that if
$S'$ has a (head-lossy)
run $(p_\init, q_\init, \hat{u}, \hat{v}) \step{*}
(p_\final, q_\final, u', v')$ then it also has a faithful one.

Let us assume that $\pi$ above, of the form $C_0\step{*}C_n$, is a witness
run of $S'$, not necessarily faithful, having minimal length.
We show how to modify it locally so
that the resulting run is faithful.

Assume that some rule $\delta_i=p\step{\top}p_\theta$ is used in $\pi$, and
that P2 fails on this occurrence of $\delta_i$. Since $\pi$ does not end in
state $p_\theta$, Sender necessarily continues with some (possibly zero)
$p_\theta\step{\ttc!\ttn}p_\theta$ steps, followed by some
$\delta_j=p_\theta\step{\ttc!x}p'$. Now all Receiver or message loss steps
between $\delta_i$ and $\delta_j$ can be swapped and postponed after
$\delta_j$ since Receiver has no tests and Sender does not test between
$\delta_i$ and $\delta_j$ (recall
Lemma~\ref{lem-commuting-steps}\eqref{comm-as}). After the transformation,
$\delta_i$ and the rules after it form a P2-segment. Also, since message
losses have been postponed, the run remains head-lossy.

Consider now a rule $\delta_i$ of the form $p\step{\testc{Z}}p^1_\ttc$ in
$\pi$ and assume that P1 fails on this occurrence. Sender necessarily
continues with some $\delta_j=p^1_\ttc\step{\ttc!\ttz}p^2_\ttc$ and
$\delta_k=p^2_\ttc\step{\testc{Z}}p$, interleaved with Receiver's steps
and/or losses. It is clear that the $\ttz$ written on $\ttc$ by $\delta_j$
must be lost, or read by a Receiver's $\delta_\ell=q\step{\ttc?\ttz}q'$
before $\delta_k$ can be used. The read or loss occurs at some step $\ell$
with $j<\ell<k$. Note that Receiver does not read from $\ttc$ between
steps $i$ and $k$, except perhaps at step $\ell$.
Since Sender only tests for emptiness of $\ttc$ between
steps $i$ and $k$, all Receiver's steps and losses between steps $i$ and
$\ell$ can be swapped and put before $\delta_i$. The run remains head-lossy
since the swapped losses do not occur on $\ttc$, which is empty at step $i$.
Similarly, all non-Sender steps between steps $\ell$ and $k$ can be swapped
after $\delta_k$, preserving head-lossiness. The obtained run has a segment
of the form
$C\step{\testc{Z}}\step{\ttc!\ttz}\step{\ttc?\ttz}\step{\testc{Z}}C'$ that is
now a P1-segment, or of the form
$C\step{\testc{Z}}\step{\ttc!\ttz}\step{\los}\step{\testc{Z}}C'=C$, i.e., a
dummy loop $C\step{+}C$ that contradicts minimality of $\pi$.

\subsection{Reducing G-G-Reach[$\bm{Z_1,N_1}$] to E-G-Reach[$\bm{Z_1,N_1}$]}
\label{ssec-ggz1n1-egz1n1}

A G-G-Reach[${Z_1,N_1}$] instance where the initial contents of
$\ttr$ and $\ttl$ are restricted to (regular languages)
$U$ and $V$ respectively can be transformed into an equivalent
instance where $U$ and $V$ are both replaced with $\{\epsilon\}$.
For this, one adds a new (fresh) initial state $p_\newinit$ to
Sender, from which Sender first nondeterministically generates
some word $u\in U$, writing it on $\ttr$, then generates some word
$v\in V$, writing it on $\ttl$, and then enters $p_\init$, the original initial state. The resulting $S'$
is just $S$ with extra states and rules between $p_\newinit$ and $p_\init$ that mimic FSAs for $U$ and $V$.

Stating the correctness of this reduction has the form
\begin{gather}
\tag{$\star$}
\label{eq-correct-5.3}
\text{$S$ has a run $(p_\init,q_\init,u,v) \step{*} C$ for some $u\in
  U$ and $v\in V$ iff $S'$ has a run
  $(p_\newinit,q_\init,\epsilon,\epsilon) \step{*} C$} \:.
\end{gather}
Now, since $S'$ can do $(p_\newinit,q_\init,\epsilon,\epsilon)
\step{*} (p_\init,q_\init,u,v)$ for any $u\in U$ and $v\in V$, the
left-to-right implication in \eqref{eq-correct-5.3} is clear. Note
that, in the right-to-left direction, \emph{it is essential that
  Receiver has no tests} and this is what we missed
in~\cite{JKS-tcs2012}. Indeed, it is the absence of Receiver tests that allows us
to reorder any $S'$ run from $(p_\newinit,q,\epsilon,\epsilon)$ so
that all steps that use the new ``generating'' rules (from
$p_\newinit$ to $p_\init$) happen before any Receiver steps.

\subsection{Reducing E-G-Reach[$\bm{Z_1,N_1}$] to E-G-Reach[$\bm{Z_1}$]}
\label{ssec-egz1n1-egz1}

When there are no Receiver tests
and a run starts with the empty channels,
 then $N_1$ tests can be easily eliminated by a
buffering technique on Sender's side.
Each channel $\ttc\in\{\ttr,\ttl\}$ gets its one-letter buffer
$\textsc{B}_\ttc$, which
can be  emptied at any time by moving its content to $\ttc$.
Sender can only write to an empty buffer; it passes a $Z^\ttc_1$ test
if both channel $\ttc$ and $\textsc{B}_\ttc$ are empty, while any
$N^\ttc_1$ test is replaced with the (weaker) ``test'' if
$\textsc{B}_\ttc$ is nonempty.

Formally, we start with an instance
$(S,p_\init,p_\final,q_\init,q_\final,\{\epsilon\},\{\epsilon\},U',V')$
of E-G-Reach[$Z_1,N_1$], where
$S=(\{\ttr,\ttl\},\Mess,Q_1,\Delta_1,Q_2,\Delta_2)$, and we create
$S'=(\{\ttr,\ttl\},\Mess,Q'_1,\Delta'_1,Q_2,\Delta_2)$ arising from
$S$ as follows (see Fig.~\ref{fig-simul5}).
\begin{figure}[htbp]
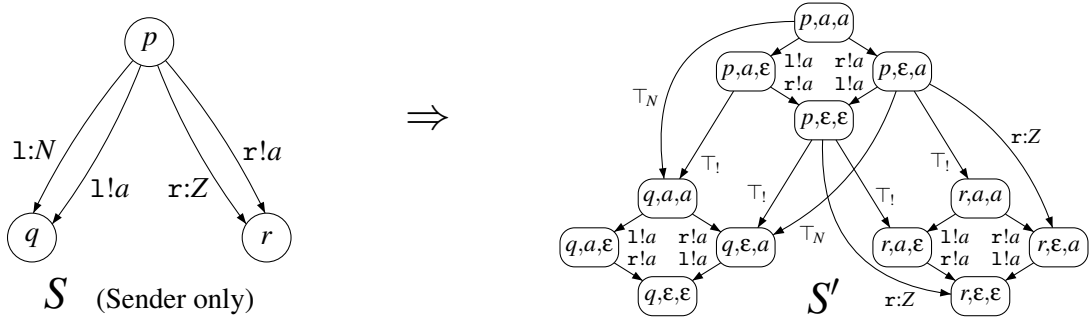

\centering
{\setlength{\unitlength}{1.3mm}
\begin{gpicture}(115,29)(20,-8)

\gasset{ilength=3.5,flength=3.5,iangle=0,fangle=0,Nw=5,Nh=5,Nmr=999,ELdist=0.5,loopdiam=3}

\node(p1)(41,18){$p$}
\node(p2)(29,-2){$q$}
\node(p3)(53,-2){$r$}
\node[Nframe=n](name)(41,-8){{\LARGE $S$}\quad (Sender only)}
\drawedge[ELside=l,ELpos=65,curvedepth=1,exo=1](p1,p2){$\ttl!a$}
\drawedge[ELside=r,ELpos=65,curvedepth=-1,exo=-1](p1,p2){$\testl{N}$}
\drawedge[ELside=l,ELpos=65,curvedepth=1,exo=1](p1,p3){$\ttr!a$}
\drawedge[ELside=r,ELpos=65,curvedepth=-1,exo=-1](p1,p3){$\testr{Z}$}

\node[Nframe=n](arrow)(69.5,10){\LARGE $\Rightarrow$}

\node[Nw=6,Nh=4,Nmr=1.5](p1)(110,20){\sizeB $p,\!a,\!a$}
\node[Nw=6,Nh=4,Nmr=1.5](p1b)(102,15){\sizeB $p,\!a,\!\epsilon$}
\node[Nw=6,Nh=4,Nmr=1.5](p1c)(118,15){\sizeB $p,\!\epsilon,\!a$}
\node[Nw=6,Nh=4,Nmr=1.5](p1d)(110,10){\sizeB $p,\!\epsilon,\!\epsilon$}
\node[Nw=6,Nh=4,Nmr=1.5](p2)(94,2){\sizeB $q,\!a,\!a$}
\node[Nw=6,Nh=4,Nmr=1.5](p2b)(86,-3){\sizeB $q,\!a,\!\epsilon$}
\node[Nw=6,Nh=4,Nmr=1.5](p2c)(102,-3){\sizeB $q,\!\epsilon,\!a$}
\node[Nw=6,Nh=4,Nmr=1.5](p2d)(94,-8){\sizeB $q,\!\epsilon,\!\epsilon$}
\node[Nw=6,Nh=4,Nmr=1.5](p3)(126,2){\sizeB $r,\!a,\!a$}
\node[Nw=6,Nh=4,Nmr=1.5](p3b)(118,-3){\sizeB $r,\!a,\!\epsilon$}
\node[Nw=6,Nh=4,Nmr=1.5](p3c)(134,-3){\sizeB $r,\!\epsilon,\!a$}
\node[Nw=6,Nh=4,Nmr=1.5](p3d)(126,-8){\sizeB $r,\!\epsilon,\!\epsilon$}
\node[Nframe=n](name)(110,-8){\LARGE $S'$}

\drawedge[ELside=l,ELpos=55](p1b,p1d){\sizeA $\ttr!a$}
\drawedge[ELside=r,ELpos=55](p1c,p1d){\sizeA $\ttl!a$}
\drawedge[ELside=r,ELpos=45](p1,p1c){\sizeA $\ttr!a$}
\drawedge[ELside=l,ELpos=45](p1,p1b){\sizeA $\ttl!a$}
\drawedge[ELside=l,ELpos=55](p2b,p2d){\sizeA $\ttr!a$}
\drawedge[ELside=r,ELpos=55](p2c,p2d){\sizeA $\ttl!a$}
\drawedge[ELside=r,ELpos=45](p2,p2c){\sizeA $\ttr!a$}
\drawedge[ELside=l,ELpos=45](p2,p2b){\sizeA $\ttl!a$}
\drawedge[ELside=l,ELpos=55](p3b,p3d){\sizeA $\ttr!a$}
\drawedge[ELside=r,ELpos=55](p3c,p3d){\sizeA $\ttl!a$}
\drawedge[ELside=r,ELpos=45](p3,p3c){\sizeA $\ttr!a$}
\drawedge[ELside=l,ELpos=45](p3,p3b){\sizeA $\ttl!a$}

\drawedge[ELdist=-0.1,ELside=l,ELpos=65,curvedepth=0](p1d,p3b){\sizeA $\top_{\!!}$}
\drawedge[ELdist=0.1,ELside=r,ELpos=65,curvedepth=0](p1c,p3){\sizeA $\top_{\!!}$}
\drawedge[ELside=r,ELpos=65,curvedepth=0](p1d,p2c){\sizeA $\top_{\!!}$}
\drawedge[ELside=l,ELpos=65,curvedepth=0](p1b,p2){\sizeA $\top_{\!!}$}
\drawedge[ELside=r,ELpos=70,curvedepth=-8](p1,p2){\sizeA $\top_{\!N}$}
\drawedge[ELside=l,ELpos=75,curvedepth=3](p1c,p2c){\sizeA $\top_{\!N}$}
\drawedge[ELside=r,ELpos=75,curvedepth=-6.5](p1d,p3d){\sizeA $\testr{Z}$}
\drawedge[ELside=l,ELpos=55,curvedepth=3](p1c,p3c){\sizeA $\testr{Z}$}

\end{gpicture}}
\caption{Reducing E-G-Reach[$Z_1,N_1$] to E-G-Reach[$Z_1$]}
\label{fig-simul5}
\end{figure}

%
\noindent We put
$Q'_1=Q_1\times(\Mess\cup\{\epsilon\})\times(\Mess\cup\{\epsilon\})$;
the components $x,y$ in a state $\inang{q,x,y}$ denote the contents of
the buffers for $\ttr$ and $\ttl$, respectively.
We now replace each rule $q\step{\ttr!x}q'$ with
$\inang{q,\epsilon,y}\step{\top}\inang{q',x,y}$ for all
$y\in\Mess\cup\{\epsilon\}$ (Fig.~\ref{fig-simul5} uses ``$\top_{\!!}$'' to highlight these transformed rules).
Each $q\step{\testr{N}}q'$ is replaced with
$\inang{q,x,y}\step{\top}\inang{q',x,y}$ for all
$x,y$ where $x\neq\epsilon$ (Fig.~\ref{fig-simul5} uses ``$\top_{\!N}$'').
Each $q\step{\testr{Z}}q'$ is replaced with
$\inang{q,\varepsilon,y}\step{\testr{Z}}\inang{q',\epsilon,y}$
(for all $y$). Analogously we replace all $q\step{\ttl!x}q'$,
 $q\step{\ttl:N}q'$, and $q\step{\ttl:Z}q'$.
 Moreover, we add the rules
 $\inang{q,x,y}\step{\ttr!x}\inang{q,\epsilon,y}$ (for $x\neq\epsilon$)
 and $\inang{q,x,y}\step{\ttl!y}\inang{q,x,\epsilon}$ (for
 $y\neq\epsilon$).
Our desired reduction is completed, by the next lemma:

\begin{lem}
	$S$ has a run
	$C_\init=(p_\init,q_\init,\epsilon,\epsilon)\step{*}
(p_\final,q_\final,u',v')=C_\final$
if, and only if, $S'$ has a
run
$C'_{\init}=(\inang{p_\init,\epsilon,\epsilon},\inang{q_\init,\epsilon,\epsilon},
\epsilon,\epsilon)\step{*}
(\inang{p_\final,\epsilon,\epsilon},
\inang{q_\final,\epsilon,\epsilon},u',v')=C'_\final$.
\end{lem}
\begin{proof}
	$\Leftarrow$ : A run
	$C'_\init=C'_0\step{\delta'_1}C'_1\step{\delta'_2}C'_2
	\cdots \step{\delta'_n}C'_n=C'_\final$
of $S'$ can be simply translated
to a run of $S$ by the following transformation:
each $C'_i=(\inang{p_i,x,y},q_i,u_i,v_i)$
is translated to $C_i=(p_i,q_i,u_ix,v_iy)$,
each step $C'_{i-1}\step{\delta'_{i}}C'_{i}$ where $\delta'_{i}$
is
$\inang{q,\epsilon,y}\step{\top}\inang{q',x,y}$
is replaced with $C_{i-1}\step{\delta}C_{i}$
where $\delta$ is
$q\step{\ttr!x}q'$, etc. It can be easily checked that
the arising run $C_0\step{*}C_n$ is indeed a valid run of $S$
(that can be shorter because it ``erases'' the steps by the rules
$\inang{q,x,y}\step{\ttr!x}\inang{q,\epsilon,y}$
 and $\inang{q,x,y}\step{\ttl!y}\inang{q,x,\epsilon}$).

$\Rightarrow$ : A run $C_\init=C_0\step{\delta_1}C_1\step{\delta_2}C_2
	\cdots \step{\delta_n}C_n=C_\final$
of $S$ can be translated
into a run of $S'$ by a suitable transformation, starting
with $C'_0=(\inang{p_\init,\epsilon,\epsilon},\inang{q_\init,\epsilon,\epsilon},
\epsilon,\epsilon)$. Suppose that
$C_0\step{*}C_i=(p,q,ux,vy)$
has been translated to
$C'_0\step{*}C'_i=(\inang{p,x,y},q,u,v)$
(for some $x,y\in\Mess\cup\{\epsilon\}$).
If $\delta_{i+1}$ is $p\step{\ttr!a}p'$, then
we translate $C_i\step{\delta_i}C_{i+1}$
in the case $x=\epsilon$
to
$C'_{i}\step{}C'_{i+1}=(\inang{p',a,y},q,u,v)$
(using the rule
$\inang{p,\epsilon,y}\step{\top}\inang{p',a,y}$), and in the case
$x\neq\epsilon$ to
$C'_{i}\step{}(\inang{p,\epsilon,y},q,ux,v)\step{}
(\inang{p',a,y},q,ux,v)=C'_{i+1}$
(using the rules $\inang{p,x,y}\step{\ttr!x}\inang{p,\epsilon,y}$
and $\inang{p,\epsilon,y}\step{\top}\inang{p',a,y}$).
We handle the other forms of $\delta_{i+1}$
 in the obvious way;
e.g., if $\delta_{i+1}$ is a loss at (the head of) $\ttl$ while
$C'_i=(\inang{p,x,y},q,u,\epsilon)$, then we also use two steps:
$C'_{i}\step{}(\inang{p,x,\epsilon},q,u,y)\step{\los}
(\inang{p,x,\epsilon},q,u,\epsilon)=C'_{i+1}$.
This process obviously results in a valid run of $S'$.
\end{proof}

\subsection{Reducing E-G-Reach[$\bm{Z_1}$] to E-E-Reach[$\bm{Z_1}$]}
\label{ssec-egz1-eez1}

The idea of the reduction is similar to what was done in
section~\ref{ssec-ggz1n1-egz1n1}. The regular final conditions
``$u'\in U'$'' and ``$v'\in V'$'' are  checked by Receiver consuming the
final channel contents.  When Sender (guesses that it) is about to write the first
message that will be part of the final $u'$ in $\ttr$ (respectively,
the final $v'$ in $\ttl$), it signals this by inserting a special
symbol $\ttsharp$ just before. After it has written $\ttsharp$ to a channel,
Sender is not allowed to test that channel anymore.

Formally we start with an instance
$(S,p_\init,p_\final,q_\init,q_\final,\{\epsilon\},
\{\epsilon\},U',V')$
of E-G-Reach[$Z_1$], where
$S=(\{\ttr,\ttl\},\Mess,Q_1,\Delta_1,Q_2,\Delta_2)$.
With $S$ we associate
$S'$
where $\Mess'=\Mess\uplus\{\ttsharp\}$,
as sketched
in Fig.~\ref{fig-simul4}.
This  yields the instance
$(S',p'_\init,p'_\final,q_\init,q_{\!f},\{\epsilon\},
\{\epsilon\},\{\epsilon\},\{\epsilon\})$
of E-E-Reach[$Z_1$], for the new final Receiver state $q_{\!f}$.
\begin{figure}[htbp]
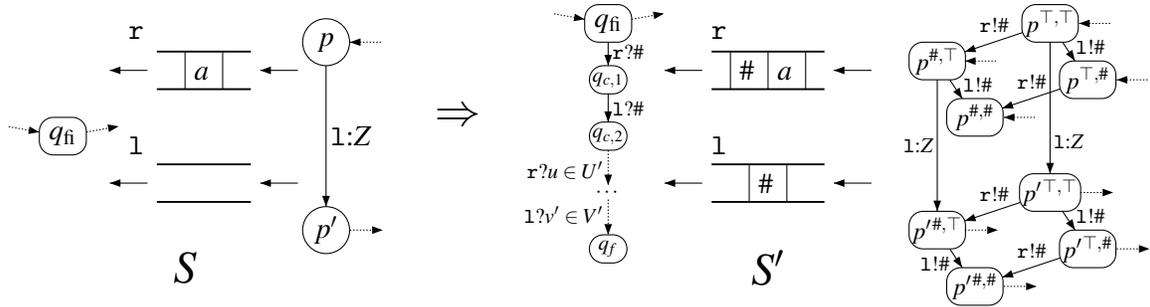

\centering
{\setlength{\unitlength}{1.25mm}
\begin{gpicture}(115,29)(3,-8)

\gasset{ilength=3.5,flength=3.5,iangle=0,fangle=0,Nw=5,Nh=5,Nmr=999,ELdist=0.5,loopdiam=3}

\node[Nh=4,Nmr=1.5](qfi)(6,8){$q_\final$}
\imark[dash={0.15 0.2}0,iangle=170](qfi)
\fmark[dash={0.15 0.2}0,fangle=10](qfi)

\node(p1)(34,18){$p$}
\node(p2)(34,-2){$p'$}
\drawedge[ELside=l](p1,p2){$\testl{Z}$}
\imark[dash={0.15 0.2}0](p1)
\fmark[dash={0.15 0.2}0](p2)

{\gasset{AHnb=0,Nframe=n}

 {\gasset{linewidth=0.2}
 \drawline(16,1)(26,1)
 \drawline(16,5)(26,5)
 \drawline(16,13)(26,13)
 \drawline(16,17)(26,17)}
 \node(name)(19,-6){\LARGE $S$}

\put(13,19){\makebox(0,0)[l]{channel $\ttr$}}
\put(13,7){\makebox(0,0)[l]{channel $\ttl$}}

 \drawline(19,13)(19,17) \put(20,14){$a$}
 \drawline(23,13)(23,17)

 \drawline[AHnb=1](15,15)(11,15)
 \drawline[AHnb=1](15,3)(11,3)
 \drawline[AHnb=1](31,15)(27,15)
 \drawline[AHnb=1](31,3)(27,3)
}

\node[Nframe=n](arrow)(48,10){\LARGE $\Rightarrow$}

\node[Nh=4,Nmr=1.5](qfi)(64,20){$q_\final$}
\imark[dash={0.15 0.2}0,iangle=170](qfi)
\fmark[dash={0.15 0.2}0,fangle=10](qfi)
\node[Nw=4,Nh=3,Nmr=1.5](qc1)(64,14){\sizeA $q_{c,1}$}
\node[Nw=4,Nh=3,Nmr=1.5](qc2)(64,8){\sizeA $q_{c,2}$}
\node[Nframe=n,w=4,Nh=1.5,Nmr=1.5](qc3)(64,2){\sizeA $\cdots$}
\node[Nw=4,Nh=3,Nmr=1.5](qc4)(64,-4){\sizeA $q_{\!f}$}
\drawedge(qfi,qc1){\sizeA $\ttr?\ttsharp$}
\drawedge(qc1,qc2){\sizeA $\ttl?\ttsharp$}
\drawedge[dash={0.15 0.2}0,ELpos=60,ELside=r](qc2,qc3){\sizeA $\ttr?u\in U'$}
\drawedge[dash={0.15 0.2}0,ELpos=40,ELside=r](qc3,qc4){\sizeA $\ttl?v'\in V'$}

\node[Nw=6,Nh=4,Nmr=1.5](p1)(111,20){\sizeB $p^{\top,\top}$}
\node[Nw=6,Nh=4,Nmr=1.5](p1b)(99,16){\sizeB $p^{\ttsharp,\top}$}
\node[Nw=6,Nh=4,Nmr=1.5](p1c)(115,14){\sizeB $p^{\top,\ttsharp}$}
\node[Nw=6,Nh=4,Nmr=1.5](p1d)(103,10){\sizeB $p^{\ttsharp,\ttsharp}$}
\node[Nw=6.5,Nh=4,Nmr=1.5](p2)(111,2){\sizeB $p'{}^{\top,\top}$}
\node[Nw=6,Nh=4,Nmr=1.5](p2b)(99,-2){\sizeB $p'{}^{\ttsharp,\top}$}
\node[Nw=6,Nh=4,Nmr=1.5](p2c)(115,-4){\sizeB $p'{}^{\top,\ttsharp}$}
\node[Nw=6,Nh=4,Nmr=1.5](p2d)(103,-8){\sizeB $p'{}^{\ttsharp,\ttsharp}$}
\imark[dash={0.15 0.2}0](p1)
\fmark[dash={0.15 0.2}0](p2)
\imark[dash={0.15 0.2}0](p1b)
\fmark[dash={0.15 0.2}0](p2b)
\imark[dash={0.15 0.2}0](p1c)
\fmark[dash={0.15 0.2}0](p2c)
\imark[dash={0.15 0.2}0](p1d)
\fmark[dash={0.15 0.2}0](p2d)

\drawedge[ELside=r,ELpos=45](p1,p1b){\sizeA $\ttr!\ttsharp$}
\drawedge[ELside=r,ELpos=45](p2,p2b){\sizeA $\ttr!\ttsharp$}
\drawedge[ELside=r,ELpos=45](p1c,p1d){\sizeA $\ttr!\ttsharp$}
\drawedge[ELside=r,ELpos=45](p2c,p2d){\sizeA $\ttr!\ttsharp$}
\drawedge[ELside=l,ELpos=62](p1,p1c){\sizeA $\ttl!\ttsharp$}
\drawedge[ELside=l,ELpos=62](p2,p2c){\sizeA $\ttl!\ttsharp$}
\drawedge[ELside=l,ELpos=62](p1b,p1d){\sizeA $\ttl!\ttsharp$}
\drawedge[ELside=r,ELpos=40](p2b,p2d){\sizeA $\ttl!\ttsharp$}
\drawedge[ELpos=70,ELdist=0.4,ELside=l](p1,p2){\sizeA $\testl{Z}$}
\drawedge[ELpos=50,ELdist=0.4,ELside=r](p1b,p2b){\sizeA $\testl{Z}$}

{\gasset{AHnb=0,Nframe=n}

 {\gasset{linewidth=0.2}
 \drawline(75,1)(87,1)
 \drawline(75,5)(87,5)
 \drawline(75,13)(87,13)
 \drawline(75,17)(87,17)}
 \node(name)(81,-6){\LARGE $S'$}

\put(75,19){\makebox(0,0)[l]{channel $\ttr$}}
\put(75,7){\makebox(0,0)[l]{channel $\ttl$}}

 \drawline(77,13)(77,17) \put(78,14){$\ttsharp$}
 \drawline(81,13)(81,17) \put(82,14){$a$}
 \drawline(85,13)(85,17)

 \drawline(79,1)(79,5) \put(80,2){$\ttsharp$}
 \drawline(83,1)(83,5)

 \drawline[AHnb=1](74,15)(70,15)
 \drawline[AHnb=1](74,3)(70,3)
 \drawline[AHnb=1](92,15)(88,15)
 \drawline[AHnb=1](92,3)(88,3)
}

\end{gpicture}}
\caption{Reducing E-G-Reach[$Z_1$] to E-E-Reach[$Z_1$]}
\label{fig-simul4}
\end{figure}

\noindent We define $S'=(\{\ttr,\ttl\},\Mess',Q'_1,\Delta'_1,Q'_2,\Delta'_2)$ with the Receiver part $Q'_2,\Delta'_2$  obtained from
$Q_2,\Delta_2$ by adding $q_{\!f}$ and other necessary states
and so called
\emph{cleaning
rules} so that
$q_{\!f}$ is reachable from $q_\final$ precisely
by sequences of read-steps
$\ttr?\ttsharp$, $\ttl?\ttsharp$,
$\ttr?a_1$, $\ttr?a_2$, $\dots$, $\ttr?a_{m_1}$,
 $\ttl?b_1$, $\ttl?b_2$, $\dots$, $\ttl?b_{m_2}$,
where $u'=a_1a_2\dots a_{m_1}\in U'$ and
 $v'=b_1b_2\dots b_{m_2}\in V'$.
(The new states and cleaning rules
mimic finite-state automata accepting $\{\ttsharp\}\cdot U'$ and
$\{\ttsharp\}\cdot V'$.)

The Sender part $Q'_1$, $\Delta'_1$ of $S'$ is obtained
from $Q_1,\Delta_1$ as follows.
We put $Q'_1 \egdef Q_1 \times \{\top,\ttsharp\}
\times \{\top,\ttsharp\}$, and $p'_\init=\inang{p_\init,\top,\top}$,
$p'_\final=\inang{p_\final,\ttsharp,\ttsharp}$.
A state $\inang{p,x,y}$ ``remembers'' if $\ttsharp$ has been already
written to $\ttr$ ($x=\ttsharp$) or not  ($x=\top$);
similarly for $\ttl$ (by $y=\ttsharp$ or $y=\top$).
For changing the status (just once for each channel),
$\Delta'_1$ contains the rules
$\inang{p,\top,y}\step{\ttr!\ttsharp}\inang{p,\ttsharp,y}$ and
$\inang{p,x,\top}\step{\ttl!\ttsharp}\inang{p,x,\ttsharp}$
for each $p\in Q_1$ and $x,y\in\{\top,\ttsharp\}$.
Moreover, any rule  $p\step{\ttc,\alpha}p'$ in $\Delta_1$ induces the
rules
$\inang{p,x,y}\step{\ttc,\alpha}\inang{p',x,y}$,
\emph{except for the rules}
$\inang{p,\ttsharp,y}\step{\testr{Z}}\dots$
and $\inang{p,x,\ttsharp}\step{\testl{Z}}\dots$
(i.e., $Z^\ttc_1$ tests are forbidden after $\ttsharp$ has been
written to $\ttc$).
The next lemma shows that the above reduction is correct.

\begin{lem}
\label{prop-ERZ12EEZ1}
$S$ has a run $(p_\init,q_\init,\epsilon,\epsilon) \step{*}
  (p_\final,q_\final,u',v')$ for some $u'\in U'$ and $v'\in V'$ if,
  and only if,
  $S'$ has a run $(\inang{p_\init,\top,\top},
  q_\init,\epsilon,\epsilon) \step{*}
  (\inang{p_\final,\ttsharp,\ttsharp},q_{\!f},\epsilon,\epsilon)$.
\end{lem}

\begin{proof}
``$\Rightarrow$'':
Suppose
$C_0=(p_\init,q_\init,\epsilon,\epsilon)\step{\delta_1}C_1\cdots\step{\delta_n}C_n=(p_\final,q_\final,u',v')$,
where $u'\in U'$, $v'\in V'$, is a run of $S$.
We first transform it into a mimicking run
$C'_0=(\inang{p_\init,\top,\top},q_\init,\epsilon,\epsilon)\step{*}C'_n=(\inang{p_\final,\ttsharp,\ttsharp},\linebreak[0] q_\final,\ttsharp
u',\ttsharp v')$.
This amounts to find some right points
for inserting two steps of the forms
$(\inang{p,\top,y},\linebreak[0] q,\linebreak[0] u,v)\step{\ttr!\ttsharp}
(\inang{p,\ttsharp,y},q,u\ttsharp,v)$
and
$(\inang{p,x,\top},q,u,v)\step{\ttl!\ttsharp}
(\inang{p,x,\ttsharp},q,u,v\ttsharp)$
(in some order).
For the first one,
if  $u'\neq\epsilon$ then
we find
the least index $i_1$ such that
$\delta_{i_1+1}$ is some $\ttr!a$
and the written occurrence of $a$ is \emph{permanent},
i.e., $C_{i_1}\step{\ttr!a}C_{i_1+1}$
is the step that actually writes the symbol occurring
at the head of $u'$ in
$C_n=(p_\final,q_\final,u',v')$;
if $u'=\epsilon$ then we find the least $i_1$ such that no
$\ttr!a$ and no $\testr{Z}$ are performed in
$C_j\step{\delta_{j+1}}C_{j+1}$ with $j\geq i_1$.
For $\ttl$ (and $v'$) we find $i_2$ analogously.
In either case, after $i_1$ (respectively, $i_2$) the channel
$\ttr$ (respectively, $\ttl$) is not tested for $\testr{Z}$.

Having
$C'_0\step{*}C'_n=(\inang{p_\final,\ttsharp,\ttsharp},
q_\final,\ttsharp
u',\ttsharp v')$, the ``cleaning rules'' are used to continue
with
$C'_n\step{*}(\inang{p_\final,\ttsharp,\ttsharp},q_{\!f},\epsilon,\epsilon)$.

``$\Leftarrow$'':
Consider a run
$C_0=(\inang{p_\init,\top,\top},q_\init,\epsilon,\epsilon) \step{*}
(\inang{p_\final,\ttsharp,\ttsharp},q_{\!f},\epsilon,\epsilon)=C_n$
of $S'$.
Since Receiver is in state $q_\init$ at the beginning  and in
$q_{\!f}$ at the end, the Receiver step sequence must be composed of two
parts: the first from $q_\init$ to $q_\final$, and the second from
$q_\final$ to $q_{\!f}$; the latter corresponds
to a sequence of cleaning (reading) rules.
The cleaning steps can be commuted after
message losses
(recall Lemma~\ref{lem-commuting-steps}\eqref{comm-al}),
and after Sender's rules
(Lemma~\ref{lem-commuting-steps}\eqref{comm-as})
since the first cleaning steps are $\ttr?\ttsharp$ and
 $\ttl?\ttsharp$ and
Sender does not test the channels after having written $\ttsharp$ on
them.

Hence we can assume that the run $C_0\step{*}C_n$ of $S'$
has the form
\[
C_0=(\inang{p_\init,\top,\top},q_\init,\epsilon,\epsilon)
\:\:\step{*}\:\:
C_m=(\inang{p_\final,\ttsharp,\ttsharp},q_\final,\ttsharp u',\ttsharp v')
\:\:\step{*}\:\:
C_n=((\inang{p_\final,\ttsharp,\ttsharp},q_\final,\epsilon,\epsilon)
\]
with only Receiver steps in $C_m\step{*}C_n$, which entails $u'\in U'$
and $v'\in V'$. If we now just ignore the two mode-changing steps in
the subrun $C_0\step{*}C_m$ (relying on the fact that $S'$ has no $N$
tests) we obtain a new run $C_0\step{*}C'_m$ with $C'_m=
(\inang{p_\final,\top,\top},q_\final, u', v')$. This new run can be
directly translated into a run $(p_\init,q_\init,\epsilon,\epsilon)
\step{*} (p_\final,q_\final,u',v')$ in $S$.
\end{proof}




\section{Reducing E-E-Reach[${Z_1}$] to G-G-Reach[${Z_1^\ttl}$]}
\label{sec-ucsz-to-zl}

We now describe an
algorithm deciding E-E-Reach[${Z_1}$] instances, assuming
a procedure deciding instances of  G-G-Reach[${Z_1^\ttl}$].
This is a Turing reduction.
The main idea is to partition a run of a UCST[$Z_1$] system into
subruns that do not use the $Z_1^\ttr$ tests (i.e., that only use the
$Z_1^\ttl$ tests) and connect them at configurations where $\ttr$ is
known to be empty.

For a UCST $S=(\{\ttr,\ttl\},\Mess,Q_1,\Delta_1,Q_2,\Delta_2)$,
we let $\Confre$ be
the subset of configurations
in which $\ttr$ is empty; they are thus of the form
$(p,q,\epsilon,v)$.
We have put
$C=(p,q,u,v)\subword C'=(p',q',u',v')$ iff
$p=p'$, $q=q'$, $u=u'$, and $v\subword v'$.
Hence $\Confre$ is a well-quasi-ordered by
$\subword$, unlike $\Conf$.

Slightly abusing terminology, we
say that a subset $W\subseteq \Confre$ is
\emph{regular} if there are some state-indexed regular languages
$(V_{p,q})_{p\in Q_1,q\in Q_2}$ in $\Reg(\Mess)$ such that
$W=\{(p,q,\epsilon,v)~|~v\in V_{p,q}\}$. Such regular subsets of $\Confre$
can be finitely represented using, e.g., regular expressions or
finite-state automata.

$W\subseteq\Confre$ is \emph{upward-closed} (in $\Confre$) if $C\in W$, $C\subword C'$ and
$C'\in \Confre$ imply $C'\in W$. It is \emph{downward-closed} if
$\Confre\setminus W$ is upward-closed. The upward-closure $\up W$ of
$W\subseteq \Confre$ is the smallest upward-closed set that contains $W$. A
well-known consequence of Higman's Lemma (see Remark~\ref{remark-higman-and-Conf-subword}) is that upward-closed and
downward-closed subsets of $\Confre$ are regular, and that upward-closed
subsets can be canonically represented by their finitely many minimal
elements.

For $W\subseteq\Confre$, we let $\Pre^*(W)\egdef\{C\in\Confre~|~\exists
D\in W: C\step{*}D\}$: note that
$\Pre^*(W)\subseteq\Confre$
by our definition.
\begin{lem}
\label{lem:generreachUCSTZ1l}
If $S$ is a UCST[$Z_1^\ttl$] system and $W$ is a regular subset of
$\Confre$, then $\Pre^*(W)$ is upward-closed; moreover, given an
oracle for G-G-Reach[$Z^\ttl_1$], $\Pre^*(W)$ is computable
from $S$ and $W$.
\end{lem}
\begin{proof}
We note that $\Pre^*(W)$ is upward-closed since $C\subword D$
is equivalent to
$D(\step{\los})^*C$, hence $D\in\Pre^*(C)$.

	We now assume that an oracle for G-G-Reach[$Z^\ttl_1$] is available, and
we construct a finite set $F \subseteq \Pre^*(W)$ whose
upward-closure $\up F$ is $\Pre^*(W)$. We build up $F$ in steps,
starting with $F_0 = \emptyset$; clearly $\up F_0 = \emptyset
\subseteq \Pre^*(W)$.
The $(i{+}1)$th iteration, starting with $F_i$,
proceeds as follows.

We put $W' \egdef \Confre \setminus \up F_i$; note that $W'$
is regular.
We check whether there exist some $C\in W'$ and $D\in W$ such that $C\step{*}D$;
this
can be decided using the oracle (it is a finite disjunction
of
G-G-Reach[$Z^\ttl_1$]
instances, obtained
by considering all
possibilities for Sender and Receiver states).
If the answer is ``no'', then $\up
F_i=\Pre^*(W)$; we then put $F=F_i$ and we are done.

Otherwise, the answer is ``yes'' and we look for some
concrete $C\in W'$
s.t.\ $C\step{*}D$ for some $D\in W$.
This can be done by enumerating all
$C\in W'$ and by using the oracle for G-G-Reach[$Z^\ttl_1$] again.
We are bound to eventually find such a $C$ since
$W'\cap\Pre^*(W)$ is not empty.

Once some $C$ is found, we set $F_{i+1}\egdef F_i\cup\{C\}$. Clearly
$F_{i+1}$, and so $\up F_{i+1}$, is a subset of $\Pre^*(W)$. By
construction, $\up F_0\subsetneq \up F_1\subsetneq
\up F_2\subsetneq\cdots$ is a strictly increasing sequence
of upward-closed sets. By the well-quasi-ordering property, this sequence
cannot be extended indefinitely:  eventually we will have
$\up F_i=\Pre^*(W)$, signalled by the answer ``no''.
\end{proof}

\begin{lem}
\label{lem-reduc-z1-to-z1l}
E-E-Reach[$Z_1$] is Turing reducible to G-G-Reach[$Z^\ttl_1$].
\end{lem}
\begin{proof}
	Assume $S=(\{\ttr,\ttl\},\Mess,Q_1,\Delta_1,Q_2,\Delta_2)$
	is a UCST[$Z_1$], and we ask
	if there is a run
	$C_\init=(p_\init, q_\init,\epsilon,\epsilon) \step{*}
	(p_\final, q_\final,\epsilon,\epsilon)=C_\final$. By
$S'$ we denote the  UCST[$Z_1^\ttl$] system
arising from $S$ by removing all $Z_1^\ttr$
rules.
Hence
Lemma~\ref{lem:generreachUCSTZ1l} applies to $S'$.
The set of configurations of $S$ and $S'$ is the same,
so there is no ambiguity in using the notation
$\Conf$ and $\Confre$.

We aim at computing $\Pre^*(\{C_\final\})$ for $S$.
For $k \geq 0$, let $T_k\subseteq\Confre$
be the set of $C\in\Confre$
for which there is a run $C \step{*} C_\final$ of $S$ with at most
$k$ steps that
are $Z_1^\ttr$ tests; hence $\up \{C_\final\}\subseteq T_0$ (by
message losses).
For each $k$, $T_k$ is upward-closed and $T_k
\subseteq T_{k+1}$. Defining $T = \bigcup_{k\in\Nat} T_k$, we note
that $C_\init
\step{*} C_\final$ iff $C_\init \in T$. Since $\Confre$ is well
quasi-ordered, the sequence $T_0\subseteq T_1\subseteq
T_2\subseteq\cdots$
eventually
stabilizes; hence there is $n$ such that $T_n = T_{n+1}$, which
implies that $T_n = T$.

By Lemma~\ref{lem:generreachUCSTZ1l}, and using an oracle for
G-G-Reach[$Z_1^\ttl$], we can compute
$\Pre^*_{S'}(\{C_\final\})$,
where the ``$S'$'' subscript indicates that we consider runs in $S'$, not
using $Z_1^\ttr$ tests.
Hence
$T_0 = \Pre^*_{S'}(\{C_\final\})$ is computable. Given $T_k$, we compute $T_{k+1}$ as follows.
We put
\begin{align*}
T'_k &=\{C\in\Confre~|~\exists D\in T_k: C\step{\testr{Z}}D\}
\\
&= \{(p, q, \epsilon, w) ~|~ \exists p' \in Q_1:  p \step{\testr{Z}} p'
\in \Delta_1 \text{ and } (p', q, \epsilon, w) \in T_k
\}
\:.
\end{align*}
Thus $T'_k \subseteq \Confre$ is the set of configurations from which one
can reach $T_k$ with one $Z_1^\ttr$ step. Clearly $T'_k$ is
upward-closed (since $T_k$ is) and can be computed from a finite representation of $T_k$,
e.g., its minimal elements. Then $T_{k+1} = T_k \cup \Pre^*_{S'}(T'_k)$,
and we use Lemma~\ref{lem:generreachUCSTZ1l} again to compute it.

Iterating the above process, we compute the sequence $T_0,
T_1, \ldots$, until the first $n$ such that $T_n = T_{n+1}$ (recall
that $T_n=T$ then).
Finally we check if $C_\init \in T_n$.
\end{proof}


\section{Reducing E-E-Reach[${Z_1^\ttl}$] to a Post Embedding Problem}
\label{sec-ucszl}

As stated in Theorem~\ref{thm-ZN-reduces-to-Z1} (see also Fig.~\ref{fig-roadmap}),
our series of reductions from  G-G-Reach[$Z_1, N_1$] to
E-E-Reach[$Z_1$] also reduces G-G-Reach[$Z_1^\ttl$] to
E-E-Reach[$Z_1^\ttl$]; this can be easily checked by recalling that
the respective
reductions do not introduce new tests.
In Subsection~\ref{subsec-ucszl-one}
we show
a (polynomial)
many-one reduction from E-E-Reach[$Z_1^\ttl$] to $\PEPpcod$, a generalization of Post's Embedding
Problem.
Since $\PEPpcod$ was shown
decidable
in~\cite{KS-msttocs}, our proof of Theorem~\ref{thm-main} will be thus
completed.
We also add
Subsection~\ref{subsec-ucszl-two} that shows a simple
reduction in the opposite direction, from $\PEPpcod$ to E-E-Reach[$Z_1^\ttl$].

\subsection{E-E-Reach[$\bm{Z_1^\ttl}$] reduces to $\bm{\PEPpcod}$}
\label{subsec-ucszl-one}

\begin{defi}[Post embedding with partial codirectness~\cite{KS-msttocs}]
$\PEPpcod$ is the question, given two finite alphabets $\Sigma,\Gamma$,
  two morphisms $u,v:\Sigma^*\to\Gamma^*$, and two regular
  languages $R,R'\in\Reg(\Sigma)$, whether there is $\sigma\in R$ (called
  a \emph{solution}) such
  that $u(\sigma)\subword v(\sigma)$, and such that furthermore
  $u(\sigma')\subword v(\sigma')$ for all suffixes $\sigma'$ of $\sigma$ that
  belong to $R'$.
\end{defi}

The above definition uses the same subword relation, denoted $\subword$, that
captures message losses. $\PEPpcod$ and $\PEP$ (which is the special
case where $R'=\emptyset$) are a variant of Post's
Correspondence Problem, where the question is whether there exists
$\sigma\in\Sigma^+$ such that $u(\sigma)=v(\sigma)$; see
also~\cite{barcelo2013} for applications in graph logics.

\begin{lem}
\label{lem-z1l2pep}
E-E-Reach[$Z_1^\ttl$] reduces to $\PEPpcod$ (via a polynomial
reduction).
\end{lem}

We now prove the lemma.
The reduction from E-E-Reach[$Z_1^\ttl$] to $\PEPpcod$ extends
an earlier reduction from UCS to $\PEP$~\cite{CS-omegapep}. In our case the
presence of $Z_1^\ttl$ tests creates new difficulties.

We fix an instance $S=(\{\ttr, \ttl\}, \Mess, Q_1, \Delta_1, Q_2,
\Delta_2)$, $C_\init = (p_\init, q_\init, \epsilon, \epsilon)$,
$C_\final = (p_\final, q_\final, \epsilon, \epsilon)$
of E-E-Reach[$Z_1^\ttl$],
and we construct a $\PEPpcod$ instance $\Pcal=(\Sigma,\Gamma,u,v,R,R')$
intended to express the existence of a run from $C_\init$ to $C_\final$.

We first put $\Sigma\egdef\Delta_1\cup\Delta_2$ and $\Gamma\egdef\Mess$ so
that words $\sigma\in\Sigma^*$ are sequences of rules of $S$, and their
images $u(\sigma),v(\sigma)\in\Gamma^*$ are sequences of messages. With any
$\delta\in\Sigma$, we associate $\wr(\delta)$
defined by $\wr(\delta)=x$ if $\delta$ is a Sender rule
of the form $p\step{\ttr!x}p'$, and $\wr(\delta)=\epsilon$ in all other
cases. This is extended to sequences with
$\wr(\delta_1\cdots\delta_n)=\wr(\delta_1)\cdots\wr(\delta_n)$. In a
similar way we define $\wl(\sigma)\in\Mess^*$, the message
sequence written to
$\ttl$ by the rule sequence $\sigma$,
and $\rr(\sigma)$ and $\rl(\sigma)$, the sequences read by $\sigma$ from
$\ttr$ and $\ttl$, respectively.
We define $E_\ttr\in\Reg(\Sigma)$ as $E_\ttr\egdef E_1\cup E_2$ where
\begin{align*}
E_1 \egdef
& \{\delta\in\Sigma~|~\wr(\delta)=\rr(\delta)=\epsilon\}      \:,
\\
E_2 \egdef
& \{\delta_1\delta_2\in\Sigma^2~|~\wr(\delta_1)=\rr(\delta_2)\not=\epsilon\} \:.
\end{align*}
In other words, $E_1$ gathers the rules that do not write to or read from
$\ttr$, and $E_2$ contains all pairs of Sender/Receiver rules that write/read
the same letter
to/from $\ttr$.

Let now $P_1 \subseteq \Delta_1^*$ be the set of all sequences of Sender
rules of the form
$p_\init=p_0\step{..}p_1\step{..}p_2\cdots\step{..}p_n=p_\final$, i.e.,
the sequences corresponding to paths from $p_\init$ to
$p_\final$ in the graph defined by $Q_1$ and $\Delta_1$.
Similarly, let $P_2\subseteq \Delta_2^*$ be the set of all sequences of
Receiver rules that correspond to paths
from $q_\init$ to $q_\final$.
Since $P_1$ and $P_2$ are defined by finite-state systems, they
are regular languages. We write $P_1 \pjshuffle P_2$ to denote the set of all
interleavings (shuffles) of a word in $P_1$ with a word in $P_2$. This
operation is regularity-preserving, so $P_1 \pjshuffle P_2\in\Reg(\Sigma)$.
Let $T_\ttl \subseteq \Delta_1$ be the set of all Sender rules that test
the emptiness of $\ttl$
(which are the only test rules in $S$).
We define $R$ and $R'$ as the following regular
languages:
\begin{xalignat*}{2}
R&=E_\ttr^* \cap(P_1\pjshuffle P_2),
&
R'&= T_\ttl\cdot\bigl(\Delta_1\cup\Delta_2\bigr)^*.
\end{xalignat*}
Finally, the morphisms $u,v:\Sigma^*\rightarrow\Gamma^*$
are given by $u\egdef\rl$ and $v\egdef\wl$.
This finishes the construction of the $\PEPpcod$ instance $\Pcal=(\Sigma,\Gamma,u,v,R,R')$.

\medskip

We will now prove the correctness of this reduction, i.e., show that
$S$ has a run $C_\init\step{*}C_\final$ if, and only if, $\Pcal$ has a
solution. Before starting with the proof itself, let us illustrate
some aspects of the reduction by considering a schematic example (see
Fig.~\ref{fig-ex-for-pep-reduction}).
\begin{figure}[htbp]
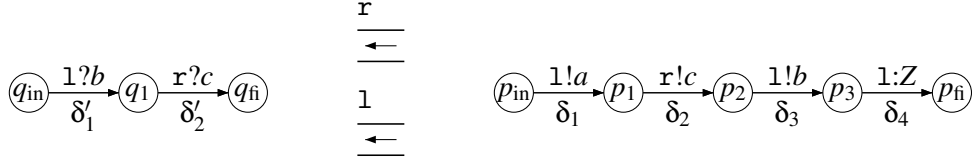

\centering
{\setlength{\unitlength}{1.04mm}
\begin{gpicture}(124,24)(-12,-2)

\gasset{ELside=r,ilength=4,flength=4,Nw=5,Nh=5,Nmr=999,loopdiam=5}

\node(q1)(-2,10){$q_\init$}
\node(q2)(12,10){$q_1$}
\node(q3)(26,10){$q_\final$}
\drawedge(q1,q2){$\delta'_1$}
\drawedge(q2,q3){$\delta'_2$}
{\gasset{ELside=l}
 \drawedge(q1,q2){$\ttl?b$}
 \drawedge(q2,q3){$\ttr?c$}
}

\node(p1)(60,10){$p_\init$}
\node(p2)(74,10){$p_1$}
\node(p3)(88,10){$p_2$}
\node(p4)(102,10){$p_3$}
\node(p5)(116,10){$p_\final$}
\drawedge(p1,p2){$\delta_1$}
\drawedge(p2,p3){$\delta_2$}
\drawedge(p3,p4){$\delta_3$}
\drawedge(p4,p5){$\delta_4$}
{\gasset{ELside=l}
 \drawedge(p1,p2){$\ttl!a$}
 \drawedge(p2,p3){$\ttr!c$}
 \drawedge(p3,p4){$\ttl!b$}
 \drawedge(p4,p5){$\testl{Z}$}
}

{\gasset{linewidth=0.2}
\drawline[AHnb=0](40,2)(46,2)
\drawline[AHnb=0](40,6)(46,6)
\drawline[AHnb=0](40,14)(46,14)
\drawline[AHnb=0](40,18)(46,18)}
\drawline(45,16)(41,16)
\drawline(45,4)(41,4)
\put(40,20){$\ttr$}
\put(40,8){$\ttl$}

\end{gpicture}}
\caption{A schematic UCST[$Z_1^\ttl$] instance}
\label{fig-ex-for-pep-reduction}
\end{figure}

\noindent Let us consider $\sigma_\sol = \delta_1 \delta'_1 \delta_2 \delta'_2
\delta_3 \delta_4$ and check whether it is a solution of the $\Pcal$
instance  obtained by our reduction. For this, one first checks that
$\sigma_\sol\in R$, computes $u(\sigma_\sol)=\rl(\sigma_\sol)=b$ and
check that $b\subword v(\sigma_\sol)=\wl(\sigma_\sol)=ab$. There
remains to check the suffixes of $\sigma_\sol$ that belong to $R'$,
i.e., that start with a $\testl{Z}$ rule. Here, only
$\sigma'=\delta_4$ is in $R'$,
and indeed $u(\sigma')=\epsilon\subword v(\sigma')$. Thus
$\sigma_\sol$ is a solution.

However, a solution like $\sigma_\sol$ does not directly correspond to
a run of $S$. For instance, any run $C_\init\step{*}C_\final$ in the
system from Fig.~\ref{fig-ex-for-pep-reduction} must use $\delta_3$
(write $\ttb$ on $\ttl$) before $\delta'_1$ (read it).

Reciprocally,  a run $C_\init\step{*}C_\final$ does not directly lead
to a solution. For example, on the same system the following run
\begin{equation}
\tag{$\pi$}
C_\init\step{\delta_1}C_1
\step{\delta_2}C_2\step{\delta_3}C_3=(p_3,q_\init,c,ab)
\step{\los}C_4=(p_3,q_\init,c,b)\step{\delta'_1}C_5\step{\delta_4}C_6\step{\delta'_2}C_\final
\:
\end{equation}
has an action in ``$C_3\step{\los}C_4$'' that is not accounted for in
$\Sigma$ and cannot appear in solutions of $\Pcal$. Also, the
$\Sigma$-word $\sigma_\pi= \delta_1 \delta_2 \delta_3 \delta'_1
\delta_4 \delta'_2$ obtained from $\pi$ is not a solution. It belongs
to $P_1\pjshuffle P_2$ but not to $E_\ttr^*$ (which requires that each
occurrence of $\delta_2$ is immediately followed by some
$.\step{\ttr?c}.$ rule). Note that $\sigma_\sol$ had $\delta_2$
followed by $\delta'_2$, but it is impossible in a run
$C_\init\step{*}C_\final$ to have $\delta_2$ immediately followed by
$\delta'_2$.

With these issues in mind, we introduce a notion bridging the difference between
runs of $S$ and solutions of $\Pcal$. We call $\sigma\in (\Delta_1\cup\Delta_2)^*$ a
\emph{pre-solution} if the following five conditions hold:
\begin{enumerate}[label=(c\arabic*)]
\item
$\sigma\in P_1\pjshuffle P_2$;
\item
$\rr(\sigma)=\wr(\sigma)$;
\item
$\rr(\sigma_1)$ is a prefix of $\wr(\sigma_1)$ for each prefix
$\sigma_1$
of $\sigma$;
\item
$\rl(\sigma)\subword\wl(\sigma)$;
\item
$\rl(\sigma_2)\subword \wl(\sigma_2)$ for each factorization $\sigma=\sigma_1 \delta
  \sigma_2$ where $\delta\in T_\ttl$ (i.e., $\delta$ is a $\testl{Z}$ rule).
\end{enumerate}
A pre-solution $\sigma$ has a \emph{Receiver-advancing switch} if
$\sigma=\sigma_1\delta\delta'\sigma_2$ where $\delta$ is a Sender rule, $\delta'$ is
a Receiver rule, and $\sigma'=\sigma_1\delta'\delta\sigma_2$ is again a
pre-solution. A \emph{Receiver-postponing switch} is defined analogously, for
$\delta$ being a Receiver rule and $\delta'$ being a Sender rule.
For example, the sequence $\sigma_\pi$ above is a pre-solution. It has a
Receiver-advancing switch on $\delta_3$ and $\delta'_1$, and one on $\delta_4$ and
$\delta'_2$. Note that when
$\sigma$ is a pre-solution, checking whether a potential Receiver-advancing or
Receiver-postponing switch leads again to a pre-solution only requires checking (c3)
or, respectively, (c5). Considering another example, $\sigma_\sol$, being a solution
is a pre-solution. It has two Receiver-postponing switches but only one
Receiver-advancing switch since switching $\delta_2$ and $\delta'_2$ does not
maintain (c3).

It is obvious that if there is a pre-solution $\sigma$ then there is an
\emph{advance-stable pre-solution} $\sigma'$, which means that $\sigma'$ has no
Receiver-advancing switch; there is also a \emph{postpone-stable pre-solution}
$\sigma''$ which has no Receiver-postponing switch.

\begin{claim}
\label{claim-advance-stable-to-sol}
Any  advance-stable pre-solution $\sigma$
is in $E_\ttr^*$, and it is thus a solution of $\Pcal$.
\end{claim}
\begin{proof}
Let us write an advance-stable pre-solution $\sigma$ as $\sigma_1\sigma_2$ where
$\sigma_1$ is the longest prefix such that $\sigma_1\in E_\ttr^*$; hence
$\rr(\sigma_1)=\wr(\sigma_1)$ by the definition of $E_r=E_1\cup E_2$. Now suppose
$\sigma_2\neq\epsilon$. Then $\sigma_2=\delta_1\delta_2\cdots\delta_k$ where
$\delta_1\not\in E_1$. Since $\rr(\sigma_1)=\wr(\sigma_1)$, $\delta_1$ must be of
the form $.\step{\ttr!x}.$ to guarantee (c3). Let us pick the smallest $\ell$ such
that $\delta_\ell=.\step{\ttr?x}.$ ---which must exist by (c2)--- and note that
$\ell>2$ since $\delta_1\delta_2\not\in E_2$ by maximality of $\sigma_1$. If we now
pick the least $j$ in $\{1,\ldots, \ell{-}1\}$ such that $\delta_j$ is a Sender rule
and $\delta_{j+1}$ is a Receiver rule, then switching $\delta_j$ and $\delta_{j+1}$
leads again to a pre-solution as can be checked by inspecting (c1--c5). This
contradicts the assumption that $\sigma$ is an advance-stable pre-solution.
\end{proof}
\begin{claim}
\label{claim-postpone-stable-to-run}
If $\sigma=\delta_1\ldots\delta_n$ is a postpone-stable pre-solution, $S$ has a run
of the form
$C_\init\step{\delta_1}\step{\los^*}\cdots\step{\delta_n}\step{\los^*}C_\final$.
\end{claim}

\begin{proof}
Assume that we try to fire $\delta_1,\ldots,\delta_n$ in that order, starting from
$C_\init$, and sometimes inserting message losses. Since $\sigma$ belongs to
$P_1\pjshuffle P_2$, we can only fail because at some point the current channel
contents does not allow the test or the read action carried by the next rule to be
fired, i.e., not because we end up in a control state that does not carry the next
rule.

So let us consider channel contents, starting with $\ttr$. For $i=0,\ldots,n$, let
$x_i=\rr(\delta_1\ldots\delta_i)$ and $y_i=\wr(\delta_1\ldots\delta_i)$. Since
$\sigma$ satisfies (c3), $y_i$ is some $x_ix'_i$ (and $x'_0=\epsilon$). One can
easily verify by induction on $i$ that after firing $\sigma_1\ldots\sigma_i$ from
$C_\init$, $\ttr$ contains exactly $x'_i$. In fact (c3) implies that if
$\delta_{i+1}$ reads on $\ttr$, it must read the first letter of $x'_i$ (and
$\delta_{i+1}$ cannot be a read on $\ttr$ when $x'_i=\epsilon$).

Now, regarding the contents of $\ttl$, we can rely on (c4) and conclude that the
actions in $\sigma$ write on $\ttl$ everything that they (attempt to) read, but we
do not know that messages are written \emph{before} they are needed for reading,
i.e., we do not have an equivalent of (c3) for $\ttl$. For this, we rely on the
assumption that $\sigma$ is postpone-stable. Write $\sigma$ under the form $\sigma_0
z_1\sigma_1 z_2\sigma_2 \ldots z_k\sigma_k$ where the $z_i$'s are the test rules
from $T_\ttl$, and where the $\sigma_i$'s factors contain no test rules. Note that,
inside a $\sigma_i$, all Sender rules occur before all Receiver rules thanks to
postpone-stability.

We claim that $\rl(\sigma_i)\subword\wl(\sigma_i)$ for all
$i=0,\ldots,k$: assume, by way of contradiction, that $\rl(\sigma_i)\not\subword\wl(\sigma_i)$ for some $i\in\{0,\ldots,k\}$
and let $\delta$ be the last rule in $\sigma_i$. Necessarily $\delta$ is a reading
rule. Now (c4) and (c5) entail $i<k$ and 
\[\rl(\sigma_{i} z_{i+1} \sigma_{i+1}
\ldots\sigma_k)\subword\wl(\sigma_{i} z_{i+1} \sigma_{i+1} \ldots\sigma_k)\:.\]
Then
$\rl(\sigma_i)\not\subword \wl(\sigma_i)$ entails
\begin{gather}
\tag{$\star\star$}
\label{eq-before-switch}
 \rl(\delta z_{i+1}\sigma_{i+1}\ldots z_k\sigma_k)\subword\wl(\sigma_{i+1}\ldots z_k\sigma_k)\:.
\end{gather}
There is now a Receiver-postponing switch since \eqref{eq-before-switch} ensures
that (c5) holds after switching $\delta$ and $z_{i+1}$, which contradicts the
assumption that $\sigma$ is postpone-stable.

Now, with $\rl(\sigma_i)\subword\wl(\sigma_i)$, it is easy to build a run
$C_\init\step{\delta_1}\step{\los^*}\cdots\step{\delta_n}\step{\los^*}C_\final$ and
guarantee that $\ttl$ is empty before firing any $z_i$ rule.
\end{proof}

We now see that our reduction is correct. Indeed, if $C_\init\step{\sigma}C_\final$
is a run of $S$ then $\sigma$ with all occurrences of $\los$ removed
is a pre-solution; and there is also an advance-stable
pre-solution, i.e., a solution of $\Pcal$. On the other hand, if $\sigma$ is a
solution of $\Pcal$ then $\sigma$ is a pre-solution, and there is also a
postpone-stable pre-solution, which corresponds to a run $C_\init\step{*}C_\final$
of $S$.
This finishes the proof of Lemma~\ref{lem-z1l2pep}, and of Theorem~\ref{thm-main}.

\subsection{$\bm{\PEPpcod}$ reduces to E-E-Reach[$\bm{Z_1^\ttl}$]}
\label{subsec-ucszl-two}

We now prove a converse of Lemma~\ref{lem-z1l2pep}, thus showing that
$\PEPpcod$ and E-E-Reach[$Z_1^\ttl$] are equivalent problems.
Actually, $\PEPpcod$ can be easily reduced to E-E-Reach[$Z_i^\ttc$] for any
$i\in\{1,2\}$ and $\ttc\in\Ch$, but we only show a reduction for $i=1$
and $\ttc=\ttl$ explicitly. (The other reductions would be analogous.)
\begin{lem}\label{lem-pepz1l}
$\PEPpcod$ reduces to  E-E-Reach[$Z_1^\ttl$] (via a polynomial
reduction).
\end{lem}
\begin{proof}
Given a $\PEPpcod$-instance $(\Sigma,\Gamma,u,v,R,R')$, we construct
a UCST$[Z^\ttl_1]$ system (denoted $S$) with distinguished states
$p_\init,p_\final,q_\xloop$, such that
\begin{equation}
\label{eq:soliffrun}
\tag{$\star\star\star$}
\textnormal{the instance has a solution iff $S$ has a run }
(p_\init,
  q_\xloop,\epsilon,\epsilon)\step{*}(p_\final,q_\xloop,\epsilon,\epsilon)
\:.
\end{equation}
The idea is simple:
Sender nondeterministically guesses a solution $\sigma$, writing $u(\sigma)$ on
$\ttr$ and $v(\sigma)$ on $\ttl$, and Receiver validates it, by reading identical
sequences from $\ttr$ and $\ttl$ (some messages from $\ttl$ might be lost).
We now make this idea more precise.

Let $M$ and $M'$ be deterministic FSAs recognizing $R$ and the
\emph{complement of} $R'$, respectively. Sender stepwise nondeterministically
generates $\sigma=a_1a_2\dots,a_m$, while taking the ``commitment''
that $\sigma$ belongs to $R$; concretely, after generating
$a_1a_2\dots a_i$ Sender also remembers the state reached by $M$ via
$a_1a_2\dots a_i$, and Sender cannot enter $p_\final$ when the current
state of $M$ is non-accepting. Moreover, for each $i\in\{1,2,\dots,m\}$, i.e., at
every step, Sender might decide to take a further commitment, namely that
$a_ia_{i+1}\dots,a_m\not\in R'$; for each such commitment Sender
starts a new copy of $M'$, remembering the states visited by $M'$ via
$a_ia_{i+1}\dots a_m$, and it cannot enter $p_\final$ if a copy of
$M'$ is in a non-accepting state. Though we do not bound the number of copies of
$M'$, it suffices to remember just a bounded information, namely the set of current
states of all these copies.

When generating $a_i$, Sender writes $u(a_i)$ on $\ttr$ and $v(a_i)$
on $\ttl$. To check that $\ttr$ contains a subword of $\ttl$, Receiver behaves as in
Fig.~\ref{fig-simul-thue} (that illustrates another reduction).
So far we have guaranteed that there is a run $(p_\init,
q_\xloop,\epsilon,\epsilon)\step{*}(p_\final,q_\xloop,\epsilon,\epsilon)$ iff there
is $\sigma=a_1a_2\dots,a_m\in R$ such that $u(\sigma) \subword
v(\sigma)$ (using the lossiness of $\ttl$ where $v(\sigma)$ has been written).

We finish by adding a modification guaranteeing
$u(a_ia_{i+1}\dots,a_m) \subword
v(a_ia_{i+1}\dots,a_m)$ for each $i\in\{1,2,\dots,m\}$ where Sender
does not commit to $a_ia_{i+1}\dots,a_m\not\in R'$. For such steps,
and before writing $u(a_i)$ and $v(a_i)$, Sender must simply wait until
$\ttl$ is empty, i.e., Sender initiates step $i$ by
(nondeterministically) either committing to
$a_ia_{i+1}\ldots, a_m\not\in R'$ or by taking a $Z^\ttl_1$-step.

It is now a routine exercise to verify that \eqref{eq:soliffrun} holds.
\end{proof}

\begin{rem}[On complexity]
Based on known results on the complexity of $\PEPpcod$
(see~\cite{SS-icalp11,KS-msttocs,KS-fossacs2013}), our reductions prove that
reachability for UCST[$Z,N$] is $\textbf{F}_{\omega^\omega}$-complete, using the
ordinal-recursive complexity classes introduced in~\cite{schmitz2013}.
\qed
\end{rem}


\section{Two undecidable problems for UCST[${Z,N}$]}
\label{sec-undec2}

The main result of this article is Theorem~\ref{thm-main}, showing
the decidability of the reachability
problem for UCST[$Z,N$].
In this section we argue that the emptiness and non-emptiness tests
(``$Z$'' and ``$N$'') strictly increase the expressive power of UCSes.
We do this by computational arguments, namely by exhibiting two
variants of the reachability problem that are undecidable for UCST[$Z,N$].
Since these variants are known to be decidable for plain UCSes (with no
tests), we conclude that there is no effective procedure to transform a
UCST[$Z,N$] into an equivalent UCS in general.
Subsection~\ref{subsec-recurreach} deals with the problem of \emph{recurrent
reachability} of a control state. In Subsection~\ref{subsec-writelossy} we
consider the usual reachability problem but we assume that \emph{messages can be lost
only during writing} to $\ttl$ (i.e., we assume that channel $\ttl$ is
reliable and that the unreliability is limited to the writing operation).

\subsection{Recurrent reachability}\label{subsec-recurreach}

The \emph{Recurrent Reachability Problem} asks, when given $S$ and its
states $p_\init, q_\init, p,q$, whether $S$ has an \emph{infinite} run
$C_\init=(p_\init,q_\init,\epsilon,\epsilon) \step{*}(p,q,u_1,v_1)
\step{+}(p,q,u_2,v_2) \step{+}(p,q,\ldots)\cdots$ visiting the pair $(p,q)$
infinitely often (NB: with no constraints on channel contents), called a
``\emph{$pq^\infty$-run}'' for short.

The next theorem separates UCSes from UCSTs, even from UCST[$Z_1^\ttr$],
i.e., UCSTs where the only tests are emptiness tests on $\ttr$ by Sender.
It implies that $Z_1^\ttr$ tests cannot be simulated by UCSes.
\begin{thm}
\label{thm-rec-ucsz}
Recurrent reachability is decidable for UCSes, and is $\Sigma_1^0$-complete
(hence undecidable) for UCST[$Z_1^\ttr$].
\end{thm}

We start with the upper bounds. Consider a UCST[$Z_1^\ttr$] system $S$ and
assume it admits a $pq^\infty$-run $\pi$. There are three cases:
\begin{description}
\item[case 1]
If $\pi$ uses infinitely many $Z$ tests, it
can be written under the form
\[
C_\init
\step{*} D_1\step{\testr{Z}} \step{*} (p,q,\ldots)
\step{*} D_2\step{\testr{Z}} \step{*} (p,q,\ldots)
\cdots
\step{*} D_n\step{\testr{Z}} \step{*} (p,q,\ldots)
\cdots
\]
Observe that $D_1,D_2,\ldots$ belong to $\Confre$ since they allow a
$\testr{Z}$ test. By Higman's Lemma, there exists two indexes $i<j$ such
that $D_i\subword D_j$. Then $D_j (\step{\los})^* D_i \step{*} (p,q,\ldots)
\step{*} D_j$ and we conclude that $S$ also has a ``looping''
$pq^\infty$-run, witnessed by a finite run of the form $C_\init \step{*}
(p,q,u,v) \step{+} (p,q,u,v)$.

\item[case 2]
Otherwise, if $\pi$ only uses finitely many $Z$ tests, it can be written
under the form $C_\init\step{*}C=(p,q,u,v)\step{}\cdots$ such that no test
occur after $C$. After $C$, any step by Sender can be advanced before
Receiver steps and message losses, according to
Lemma~\ref{lem-commuting-steps}\eqref{comm-as}. Assuming that $\pi$ uses
infinitely many Sender steps, we conclude that $S$ has a $pq^\infty$ run
that eventually only uses Sender rules (but no $Z$ tests). At this point,
we can forget about the contents of the channels (they are not read or
tested anymore). Hence a finite witness for such $pq^\infty$-runs is
obtained by the combination of a
finite run $C_\init\step{*}(p,q,u,v)$ and a loop
$p=p_1\step{\delta_1}p_2\step{\delta_2}\cdots p_n\step{\delta_n}p_1$ in
Sender's rules that does not use any testing rule.

\item[case 3]
The last possibility is that $\pi$ uses only finitely many Sender rules. In
that case, the contents of the channels is eventually fixed hence there is
a looping $pq^\infty$-run of the form $C_\init\step{*}C=(p,q,u,v)\step{+}C$
such that the loop from $C$ to $C$ only uses Receiver rules. A finite
witness for such cases is a finite run $C_\init\step{*}(p,q,u,v)$ combined with a
loop $q=q_1\step{\delta_1}q_2\step{\delta_2}\cdots q_n\step{\delta_n}q_1$
in Receiver's rules that only uses rules reading $\epsilon$.
\end{description}\medskip

\noindent Only the last two cases are possible for UCSes: for these systems, deciding
Recurrent reachability reduces to deciding whether some $(p,q,...)$ is
reachable and looking for a loop (necessarily with no tests) starting from $p$ in
Sender's graph, or a loop with no reads starting from $q$ in Receiver's
graph.

For UCST[$Z_1^\ttr$], one must also consider the general looping ``case
1'', i.e., $\exists u,v:C_\init\step{*}(p,q,u,\linebreak[0] v)\step{+}(p,q,u,v)$. Since
reachability is decidable, this case is in $\Sigma_1^0$, as is Recurrent
reachability for UCST[$Z_1^\ttr$].
\\

Now for the lower bound. We prove $\Sigma_1^0$-hardness by a reduction from
the looping problem for semi-Thue systems.

A \emph{semi-Thue system} $T = (\Gamma, R)$ consists of a finite alphabet
$\Gamma$ and a finite set
$R\subseteq \Gamma^*\times \Gamma^*$
of \emph{rewrite rules};
we write  $\alpha \to \beta$ instead of
$(\alpha,\beta)\in R$.
The system gives rise to a \emph{one-step rewrite relation}
$\to_R
\:\subseteq \Gamma^* \times \Gamma^*$ as expected: $x \to_R y$
$\equivdef$  $x$ and $y$ can be factored as
$x = z \alpha z'$ and $y=z\beta z'$
for some rule $\alpha \to \beta$ and some
strings $z, z' \in \Gamma^*$. As usual, we write $x \step{+}_R y$ if
$x$ can be rewritten into $y$ by a nonempty sequence of steps.

We say that $T = (\Gamma, R)$ is \emph{length-preserving} if $\size{\alpha}
= \size{\beta}$ for each rule in $R$, and that it \emph{has a loop} if
there is some $x \in \Gamma^*$ such that $x \step{+}_R x$.
The following is standard
(since the one-step relation between Turing machine configurations can be
captured by finitely many length-preserving rewrite rules).
\begin{fact}\label{fact-semiThue}
The question whether a given length-preserving semi-Thue
system has a loop is $\Sigma_1^0$-complete.
\end{fact}

We now reduce the existence of a loop for length-preserving semi-Thue
systems to the recurrent reachability problem for UCST[$Z_1^\ttr$].

Let $T = (\Gamma, R)$ be a given length-preserving semi-Thue system.
We construct a UCST $S$, with message alphabet $\Mess\egdef\Gamma
\uplus \{\ttsharp\}$. The reduction is illustrated in
Fig.~\ref{fig-simul-thue}, assuming $\Gamma=\{a,b\}$. The resulting $S$
behaves as follows:
\\
%
(a) Sender starts in state $p_\init$, begins by nondeterministically sending some $y_0 \in \Gamma^*$ on
$\ttl$, then moves to state $p_\xloop$. In state $p_\xloop$, Sender
performs the following steps in succession:
\begin{enumerate}
\item check that (equivalently, wait until) $\ttr$ is empty;
\item send $\ttsharp$ on $\ttl$;
\item nondeterministically send a string $z \in \Gamma^*$ on both $\ttl$ and $\ttr$;
\item nondeterministically choose a rewrite rule $\alpha \to \beta$
	(from $R$) and send $\alpha$ on $\ttr$ and $\beta$ on $\ttl$;
\item nondeterministically send a string $z' \in \Gamma^*$ on both $\ttl$ and $\ttr$;
\item send $\ttsharp$ on $\ttr$;
\item go back to $p_\xloop$ (and repeat 1--7).
\end{enumerate}
\begin{figure}[htbp]
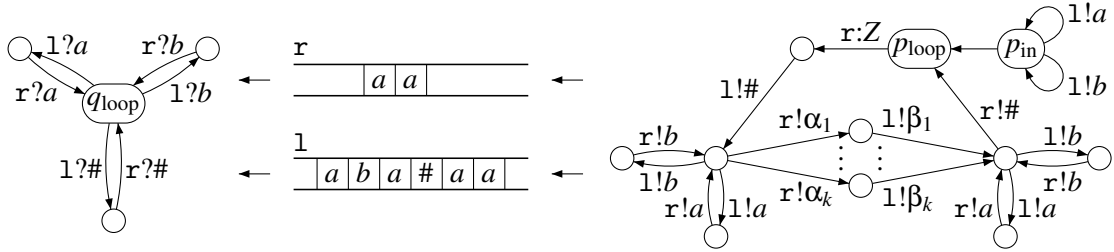

\centering
{\setlength{\unitlength}{1.04mm}
\begin{gpicture}(138,26)(4,-3)

\gasset{ilength=4,flength=4,Nw=5,Nh=5,Nmr=999,loopdiam=4}

\node[Nw=8](q1)(17,13){$q_\xloop$}
\node[Nw=3,Nh=3](q2)(5,20){}
\node[Nw=3,Nh=3](q3)(17,-2){}
\node[Nw=3,Nh=3](q4)(29,20){}
{\gasset{curvedepth=-1,ELside=r}
 \drawedge[ELpos=63,ELdist=0.3](q1,q2){$\ttl?a$}
 \drawedge[ELpos=35,ELdist=0.2](q2,q1){$\ttr?a$}
 \drawedge[ELpos=55,ELdist=0.5](q1,q3){$\ttl?\ttsharp$}
 \drawedge[ELpos=45,ELdist=0.5](q3,q1){$\ttr?\ttsharp$}
 \drawedge[ELpos=65,ELdist=0.2](q1,q4){$\ttl?b$}
 \drawedge[ELpos=35,ELdist=0.2](q4,q1){$\ttr?b$}
}

\node[Nw=6](p1)(133,20){$p_\init$}
\node[Nw=8](p2)(120,20){$p_\xloop$}
\node[Nw=3,Nh=3](p3)(105,20){}
\node[Nw=3,Nh=3](p4)(94,6){}
\node[Nw=3,Nh=3](p4a)(94,-4){}
\node[Nw=3,Nh=3](p4b)(82,6){}
\node[Nw=3,Nh=3](p5)(131,6){}
\node[Nw=3,Nh=3](p5a)(131,-4){}
\node[Nw=3,Nh=3](p5b)(143,6){}
\node[Nw=3,Nh=3](p451)(112.5,9.5){}
\node[Nw=3,Nh=3](p45k)(112.5,2.5){}
\node[Nframe=n](dots)(110,7){$\vdots$}
\node[Nframe=n](dots)(115,7){$\vdots$}
\drawloop[loopangle=45,ELpos=40,ELdist=0.5,loopCW=n,ELside=r](p1){$\ttl!a$}
\drawloop[loopangle=-45,ELpos=40,ELdist=0.5](p1){$\ttl!b$}
\drawedge(p1,p2){}
\drawedge[ELside=r](p2,p3){$\testr{Z}$}
\drawedge[ELside=r,ELdist=0.5](p3,p4){$\ttl!\ttsharp$}
\drawedge[curvedepth=1,ELpos=60,ELdist=0.5](p4,p4a){$\ttl!a$}
\drawedge[curvedepth=1,ELpos=40,ELdist=0.5](p4a,p4){$\ttr!a$}
\drawedge[curvedepth=1,ELpos=60,ELdist=0.5](p4,p4b){$\ttl!b$}
\drawedge[curvedepth=1,ELpos=40,ELdist=0.5](p4b,p4){$\ttr!b$}
\drawedge[curvedepth=1,ELpos=60,ELdist=0.5](p5,p5a){$\ttl!a$}
\drawedge[curvedepth=1,ELpos=40,ELdist=0.5](p5a,p5){$\ttr!a$}
\drawedge[curvedepth=1,ELpos=60,ELdist=0.5](p5,p5b){$\ttl!b$}
\drawedge[curvedepth=1,ELpos=40,ELdist=0.5](p5b,p5){$\ttr!b$}
\drawedge[ELpos=65,ELdist=0.5](p4,p451){$\ttr!\alpha_1$}
\drawedge[ELpos=30,ELdist=0.2](p451,p5){$\ttl!\beta_1$}
\drawedge[ELside=r,ELpos=65,ELdist=0.2](p4,p45k){$\ttr!\alpha_k$}
\drawedge[ELside=r,ELpos=30,ELdist=0.4](p45k,p5){$\ttl!\beta_k$}
\drawedge[ELside=r,ELdist=0.5,ELpos=30](p5,p2){$\ttr!\ttsharp$}

{\gasset{AHnb=0,Nframe=n}

 {\gasset{linewidth=0.2}
 \drawline(40,2)(70,2)
 \drawline(40,6)(70,6)
 \drawline(40,14)(70,14)
 \drawline(40,18)(70,18)}

\put(40,20){\makebox(0,0)[l]{channel $\ttr$ (reliable)}}
\put(40,8){\makebox(0,0)[l]{channel $\ttl$ (lossy)}}

 \drawline(43,2)(43,6) \put(44,3){$a$}
 \drawline(47,2)(47,6) \put(48,3){$b$}
 \drawline(51,2)(51,6) \put(52,3){$a$}
 \drawline(55,2)(55,6) \put(56,3){$\ttsharp$}
 \drawline(59,2)(59,6) \put(60,3){$a$}
 \drawline(63,2)(63,6) \put(64,3){$a$}
 \drawline(67,2)(67,6)

 \drawline(49,14)(49,18) \put(50,15){$a$}
 \drawline(53,14)(53,18) \put(54,15){$a$}
 \drawline(57,14)(57,18)

 \drawline[AHnb=1](37,16)(33,16)
 \drawline[AHnb=1](37,4)(33,4)
 \drawline[AHnb=1](77,16)(73,16)
 \drawline[AHnb=1](77,4)(73,4)
}

\end{gpicture}}
\caption{Solving the looping problem for semi-Thue systems}
\label{fig-simul-thue}
\end{figure}

%
\noindent The loop 1--7 above can be also summarized as: check that $\ttr$ is empty,
nondeterministically guess two strings $x$ and $y$ such that $x \to_R y$,
writing $x \ttsharp$ on $\ttr$ and $\ttsharp y$ on $\ttl$.
\\
%
(b)
Receiver starts in state $q_\xloop$ from where it
reads any pair of identical symbols from $\ttr$ and $\ttl$, returns to
$q_\xloop$, and repeats this indefinitely.

\begin{claim}[Correctness of the reduction]
$S$ has an infinite run starting from $C_\init=(p_\init, q_\xloop,
  \epsilon, \epsilon)$ and visiting the control pair $(p_\xloop, q_\xloop)$
  infinitely often if, and only if, $x \step{+}_R x$ for some $x \in
  \Gamma^*$.
\end{claim}
\begin{proof}
For the ``$\Leftarrow$'' direction we assume that $T$ has a loop $x=x_0
\to_R x_1 \to_R \ldots \to_R x_n=x$ with $n>0$. Let
$C_i\egdef(p_\xloop,q_\xloop,\epsilon,x_i)$. $S$ obviously has a run
$C_\init\step{*}C_0$, sending $x_0$ on $\ttl$. For each $i\geq
0$, $S$ has a run
$C_i\step{+}C_{i+1}$: it starts with appending the pair $x_i\to_R x_{i+1}$
on the channels, hence visiting $(.,.,x_i\,\ttsharp,x_i\,\ttsharp\, x_{i+1})$,
from which Receiver can read the $x_i\,\ttsharp$ prefix on both channels,
thus reaching $C_{i+1}$. Note that no messages are lost in these runs.
Chaining them gives an infinite run that visits $(p_\xloop,q_\xloop)$
infinitely many times.

For the ``$\Rightarrow$'' direction, we assume $S$ has an infinite run
starting from $C_\init$ that visits $(p_\xloop, q_\xloop)$ infinitely
often. Since Sender checks the emptiness of $\ttr$ before running
through its loop, we conclude that no $\ttsharp$ character written to
$\ttl$ is lost during the run.
Let $y_0$ be written on $\ttl$ before the first visit of $p_\xloop$;
for $i\geq 1$, let $(x_i, y_i)$ be the pair of strings
guessed by Sender during the $i$th iteration
of its loop 1--7
($x_i$ written on $\ttr$
and $y_i$ on $\ttl$).
Receiver can
only empty the reliable channel $\ttr$ if $x_i\subword y_{i-1}$ for all $i
\geq 1$. This implies $\size{x_i} \leq \size{y_{i-1}}$. We also have
$\size{x_i} = \size{y_i}$ since $T$ is length-preserving. Therefore
eventually, say for all $i\geq n$, all $x_i$ and $y_i$
have the same length. Then $x_i=y_{i-1}$ for $i>n$ (since
$x_i\subword y_{i-1}$ and $\size{x_i}=\size{y_{i-1}}$).
Hence $T$ admits an infinite derivation of the
form
\[
x_n \to_R y_n = x_{n+1} \to_R y_{n+1} = x_{n+2} \to_R \cdots
\]
Since there are only finitely many strings of a given length, there
are
two positions $m'>m\geq n$ such that $x_m=x_{m'}$; hence
$T$ has a loop $x_m \step{+}_R x_m$.
\end{proof}

\subsection{Write-lossy semantics}
\label{subsec-writelossy}

As another illustration of the power of tests, we consider UCSTs with
\emph{write-lossy semantics}, that is, UCSTs with the assumption that
messages are only lost during steps that write them to $\ttl$. Once
messages are in $\ttl$, they are never lost.
If we start with the empty channel $\ttl$ and we only allow the
emptiness tests on $\ttl$, then any
computation in normal lossy semantics can be mimicked by
a computation in write-lossy semantics:
any occurrence of a message that gets finally lost
will simply not be written.
Adding the non-emptiness test makes a difference, since the
reachability problem becomes undecidable.

We now make this reasoning more formal, using
the new
transition relation $C\step{}_\wrlo C'$ that is intermediary between the
reliable and the lossy semantics.

Each $\ttl$-writing rule $\delta$ of the form $p\step{\ttl!x}p'$
in a UCST $S$ will give rise to \emph{write-lossy steps} of the form
$(p,q,u,v)\step{\wrlo}(p',q,u,v)$, where $\delta$ is performed
but nothing is
actually written.
We write $C\step{}_\wrlo C'$ when there is a reliable or a
write-lossy step from $C$ to $C'$, and use $C\step{}_\rel C'$ and
$C\step{}_\los C'$ to denote the existence of a reliable step, and
respectively, of a reliable or a lossy step. Then $\step{}_\rel
\:\subseteq\: \step{}_\wrlo \:\subseteq\: \step{*}_\los$.

Now we make precise the equivalence of the two semantics
when we start with the empty $\ttl$
and only use  the emptiness tests:

\begin{lem}
\label{lem-wl-coincide}
Assume $S$ is a UCST[$Z$] system. Let $C_\init=(p,q,u,\epsilon)$ be a
configuration (where $\ttl$ is empty). Then, for any $C_\final$
configuration, $C_\init\step{*}_\los C_\final$ iff $C_\init\step{*}_\wrlo
C_\final$.
\end{lem}
\begin{proof}
The ``$\Leftarrow$'' direction is trivial.
For the ``$\Rightarrow$'' direction
we claim that
\begin{gather}
\label{eq-claim-wrlo-commute}
\tag{$\dagger$}
\text{if $C\step{}_\wrlo C'\supword_1 C''$, then also
$C\supword D\step{}_\wrlo C''$ for some $D$.}
\end{gather}
Indeed, if (the occurrence of) the
message in $C'$ that is missing in $C''$ occurs
in $C$, then it is possible to first lose this message, leading to $D$,
before mimicking the step that went from $C$ to $C'$ (we rely here on the
fact that $S$ only uses $Z$ tests). Otherwise, $C''$ is obtained by losing
the message that has just been (reliably) written when moving from $C$ to
$C'$, and taking $D=C$ is possible.

Now, since $\step{*}_\los$ is $\bigl(\step{}_\wrlo \cup \supword_1\bigr)^*$
and since $\bigl(\supword_1\bigr)^*$ is $\supword$, we can use
\eqref{eq-claim-wrlo-commute} and conclude that $C\step{*}_\los D$ implies
that $C\supword C'\step{*}_\wrlo D$ for some $C'$. Finally, in the case
where $C=C_\init$ and $\ttl$ is empty, only $C'=C_\init$ is possible.
\end{proof}
\begin{cor}
\label{coro-ucsz-wlossy}
E-G-Reachability is decidable for UCST[$Z$] with write-lossy
semantics.
\end{cor}

The write-lossy semantics is meaningful when modeling unreliability of the
writing actions as opposed to unreliability of the channels. In the
literature, write-lossy semantics is mostly used as a way of restricting
the nondeterminism of message losses without losing any essential
generality, relying on equivalences like Lemma~\ref{lem-wl-coincide} (see,
e.g., \cite[section~5.1]{CS-lics08}).

However, for our UCST systems, the write-lossy and the standard lossy
semantics do not coincide when $N$ tests are allowed. In fact,
Theorem~\ref{thm-main} does not extend to write-lossy systems.
\begin{thm}
\label{thm-ucst-wlossy-undec}
E-E-Reach is undecidable for UCST[$Z_1^\ttl,N_1^\ttl$] with write-lossy
semantics.
\end{thm}
\begin{proof}[Proof Idea]
As in Section~\ref{ssec-simulating-queue}, Sender simulates a queue automaton using tests and the help of
Receiver. See Fig.~\ref{fig-reduc2}. Channel $\ttl$ is initially empty. To
read, say, $a$ from $\ttr$, Sender does the following: (1) write $a$ on
$\ttl$; (2) check that $\ttl$ is nonempty (hence the write was not lost);
(3) check that, i.e., wait until, $\ttl$ is empty. Meanwhile, Receiver
reads identical letters from $\ttr$ and $\ttl$.
\end{proof}
\begin{figure}[htbp]
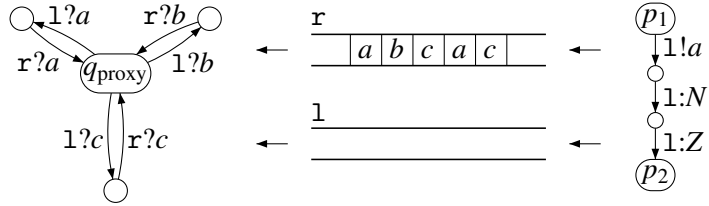

\centering
{\setlength{\unitlength}{1.04mm}
\begin{gpicture}(89,24)(4,-3)

\gasset{ilength=4,flength=4,Nw=5,Nh=5,Nmr=999,loopdiam=5}

\node[Nw=9](q1)(17,13){$q_\prx$}
\node[Nw=3,Nh=3](q2)(5,20){}
\node[Nw=3,Nh=3](q3)(17,-2){}
\node[Nw=3,Nh=3](q4)(29,20){}
{\gasset{curvedepth=-1,ELside=r}
 \drawedge[ELpos=63,ELdist=0.3](q1,q2){$\ttl?a$}
 \drawedge[ELpos=35,ELdist=0.2](q2,q1){$\ttr?a$}
 \drawedge[ELpos=55,ELdist=0.5](q1,q3){$\ttl?c$}
 \drawedge[ELpos=45,ELdist=0.5](q3,q1){$\ttr?c$}
 \drawedge[ELpos=65,ELdist=0.2](q1,q4){$\ttl?b$}
 \drawedge[ELpos=35,ELdist=0.2](q4,q1){$\ttr?b$}
}

\node[Nh=4](p1)(86,20){$p_1$}
\node[Nh=4](p2)(86,0){$p_2$}
\node[Nw=2,Nh=2](pp0)(86,13){}
\node[Nw=2,Nh=2](pp1)(86,7){}
\drawedge(p1,pp0){$\ttl ! a$}
\drawedge(pp0,pp1){$\testl{N}$}
\drawedge[ELpos=40](pp1,p2){$\testl{Z}$}

{\gasset{AHnb=0,Nframe=n}

 {\gasset{linewidth=0.2}
 \drawline(42,2)(72,2)
 \drawline(42,6)(72,6)
 \drawline(42,14)(72,14)
 \drawline(42,18)(72,18)}

\put(42,20){\makebox(0,0)[l]{channel $\ttr$ (reliable)}}
\put(42,8){\makebox(0,0)[l]{channel $\ttl$ (lossy)}}

 \drawline(47,14)(47,18) \put(48,15){$a$}
 \drawline(51,14)(51,18) \put(52,15){$b$}
 \drawline(55,14)(55,18) \put(56,15){$c$}
 \drawline(59,14)(59,18) \put(60,15){$a$}
 \drawline(63,14)(63,18) \put(64,15){$c$}
 \drawline(67,14)(67,18)

 \drawline[AHnb=1](39,16)(35,16)
 \drawline[AHnb=1](39,4)(35,4)
 \drawline[AHnb=1](79,16)(75,16)
 \drawline[AHnb=1](79,4)(75,4)
}

\end{gpicture}}
\caption{Write-lossy Sender simulates ``$p_1\step{\ttr?a}p_2$'' with $N$ and $Z$ tests and proxy Receiver}
\label{fig-reduc2}
\end{figure}

%
\noindent Thus, at least in the write-lossy setting, we can separate UCST[$Z$] and
UCST[$Z_1^\ttl,N_1^\ttl$] w.r.t.\  decida\-bility of reachability.



\section{Conclusion}
\label{sec-concl}

UCSes are communicating systems where a Sender can send messages to
a Receiver via one reliable and one unreliable, lossy, channel, but
where no direct communication is possible in the other direction. We
introduced UCSTs, an extension of UCSes where steps can be guarded by tests,
i.e., regular predicates on channel contents. This extension introduces
limited but real possibilities for synchronization between Sender and
Receiver. For example, Sender (or Receiver) may use tests to detect whether
the other agent has read (or written) some message. As a consequence, adding
tests leads to undecidable reachability problems in general. Our main
result is that reachability remains decidable when only emptiness and
non-emptiness tests are allowed. The proof goes through a series of
reductions from UCST[$Z,N$] to UCST[$Z_1^\ttl$] and finally
 to $\PEPpcod$, an extension of Post's Embedding Problem that was motivated
by the present article and whose decidability was recently proved by the last two
authors~\cite{KS-msttocs}.

These partial results do not yet provide a clear picture of what tests
on channel contents make reachability undecidable for UCSTs. At the
time of this writing, the two most pressing questions we would like to
see answered are:
\begin{enumerate}
\item
what about occurrence and non-occurrence tests, defined as
$\{O_a,NO_a~|~a\in\Mess\}$ with $O_a=\Mess^*.a.\Mess^*$ and
$NO_a=(\Mess\setminus\{a\})^*$? Such tests generalize $N$ and $Z$ tests and have been considered for channel systems used as a tool for questions on Metric
Temporal Logic~\cite{bouyer2007}.
\item
what about UCSTs with tests restricted to the lossy $\ttl$
channel? The undecidable reachability questions
in Theorem~\ref{thm-UCST-P1r-undec} all rely on tests on the reliable $\ttr$ channel.
\end{enumerate}



\bibliographystyle{alpha}
\bibliography{LMCS_2014_999_singlefile}


\end{document}